\documentclass{article}
\usepackage{amsmath,amssymb}
\usepackage{amsthm,amscd}
\usepackage[all]{xy}
\newcommand{\IZ}{\mathbb{Z}}
\newcommand{\IC}{\mathbb{C}}

\newtheorem{theorem}{Theorem}
\newtheorem{proposition}{Proposition}

\newtheorem{corollary}{Corollary}

\numberwithin{equation}{section}

\begin{document}

\title{Symmetries in the third Painlev\'e equation
arising from the modified Pohlmeyer-Lund-Regge hierarchy}
\author{Tetsuya Kikuchi \\
Graduate School of Mathematical Sciences, 
the University of Tokyo,  \\
Komaba, Meguro-ku, Tokyo 153-8914, Japan}
\date{\today}

\maketitle

\begin{abstract}

We propose a 
modification of the AKNS hierarchy
that
includes the ``modified'' Pohlmeyer-Lund-Regge (mPLR) equation.
Similarity reductions of 
this hierarchy give the second, third, and fourth Painlev\'e equations.
Especially, 
we present a new Lax representation 
and a complete description of the symmetry 
of the third Painlev\'e equation
through the similarity reduction.
We also show the relation 
between the tau-function of the mPLR hierarchy and 
Painlev\'e equations.

\end{abstract}


\section{Introduction}

Painlev\'e equations and their higher order
analogues 
are obtained by certain reductions 
of infinite-dimensional integrable systems.
This clarifies the
origin of various aspects of 
Painlev\'e equations, 
such as
the affine Weyl group symmetry, 
bilinear relations for $\tau$-functions, 
Lax formalism,
the solutions described by special polynomials.

In \cite{NYA} Noumi and Yamada proposed a
systematic description 
of nonlinear differential equations  
comprising the second, fourth and fifth Painlev\'e
equations
which 
possess $A_n^{(1)}$ affine Weyl group symmetry.
These systems 
are
obtained by similarity reductions from 
the $A_n^{(1)}$ Drinfeld-Sokolov
hierarchy
which is an
infinite-dimensional integrable system
characterized by affine Lie algebras
and their
Heisenberg subalgebras \cite{DS}, \cite{gds1}.
Also, 
Fuji and Suzuki 
derived dynamical systems,
including 
the sixth Painlev\'e equation,
by similarity reductions from the $D_{2n+2}^{(1)}$ Drinfeld-Sokolov
hierarchy \cite{FS}.
Therefore,
the affine Weyl group symmetries in the 
Painlev\'e II, IV, V, VI equations can all be derived
from integrable systems associated with affine Lie algebra.

In this paper we provide 
a natural and intrinsic 
description of the symmetry of the third Painlev\'e equation ($\mathrm{P_{III}}$)
\begin{equation}
\label{PIII}
 \frac{d^2f}{dx^2} = \frac{1}{f}\left(\frac{df}{dx}\right)^2
 -\frac{1}{x}\frac{df}{dx}
 + \frac{1}{x}(c_1f^2 + c_2) + c_3f^3 + \frac{c_4}{f}
\end{equation}
in terms of a hierarchy of soliton equations.
$\mathrm{P_{III}}$ \eqref{PIII}
has an equivalent equation
Painlev\'e III$'$
($\mathrm{P_{III'}}$) 
\cite{OkPIII}:
\begin{equation}
\label{PIII'}
 \frac{d^2y}{ds^2} = \frac{1}{y}\left(\frac{dy}{ds}\right)^2
 -\frac{1}{s}\frac{dy}{ds}
 + \frac{y^2}{4s^2}(c_3y + c_1) + \frac{c_2}{4s} + \frac{c_4}{4y},
\end{equation}
which is obtained 
by the change of the variables $s=x^2$, $y = xf$.
In what follows
we assume that $c_3c_4 \neq 0$ ($D_6^{(1)}$-type),
in which case the number of parameters 
contained in \eqref{PIII'} is two 
by means of a suitable change of scales for $y$ and $s$.
It is known that 
the transformation group of solutions of \eqref{PIII'} is
isomorphic to the affine Weyl group of type
$A_1^{(1)} \times A_1^{(1)}$ (or $B_2^{(1)}$) \cite{OkPIII}.

There are several works concerning the third Pailev\'e equation
based on the theory of integrable systems
\cite{JM3}, \cite{Joshi3}, \cite{Mas}, 
\cite{Tsuda}, \cite{WH}.
However,
there have not been any satisfactory
theories presented so far 
which
could explain 
the relationship between
Okamoto's
theory,
especially
the symmetry of the 
affine Weyl group 
based on a Hamiltonian equation and the $\tau$-function,
and the soliton equations realized
as representations of affine Lie algebras.
So
we develop 
the theory of the
``modified'' Pohlmeyer-Lund-Regge (mPLR) hierarchy 
that includes the
derivative nonlinear Schr\"odinger hierarchy
studied by Kakei and the
author \cite{KK1}, \cite{KK2},
in the same way as 
the sine-Gordon hierarchy includes the mKdV hierarchy.
The similarity reduction of the mPLR hierarchy 
gives the third Painlev\'e equation and its symmetry.

The ``modified'' Pohlmeyer-Lund-Regge (mPLR) equation is,
by definition,
the
following system of equations: 
\begin{equation}
\label{PLR-phiqr}
\left\{
\begin{aligned}
 \frac{\partial^2 q}{\partial t_1\partial \bar{t}_1}
 &= 4q + 4qr \frac{\partial q}{\partial \bar{t}_1},
 \\
 \frac{\partial^2  r}{\partial t_1\partial \bar{t}_1}
 &= 4r -4qr  \frac{\partial r}{\partial \bar{t}_1},
\end{aligned}
\right.
\end{equation}
where $q=q(t_1, \bar{t}_1)$,
$r = r(t_1, \bar{t}_1)$.
The zero-curvature representation of \eqref{PLR-phiqr} is
\begin{align}
\label{zerocurv-mPLR}
\frac{\partial B_1}{\partial \bar{t}_1} 
-\frac{\partial \bar{B}_1}{\partial t_1} 
= \bar{B}_1B_1 - B_1\bar{B}_1,
\end{align}
where
\begin{align}
\label{B1}
 B_1 &= \begin{bmatrix}  2qr & -2q \\ 0 & -2qr \end{bmatrix}
+ \begin{bmatrix}  1 & 0 \\ 2r & -1 \end{bmatrix}\zeta,
 \\
\label{B1bar}
 \bar{B}_1 &
= \begin{bmatrix}  1 & -2\bar{q}e^{2\phi} \\ 0 & -1 \end{bmatrix}\zeta^{-1}
 + \begin{bmatrix} 0 & 0 \\ 2\bar{r}e^{-2\phi} & 0 \end{bmatrix}.
\end{align}
Here $\zeta \in \IC$ is a parameter 
and the variables $\bar{q} = \bar{q}(t_1, \bar{t}_1)$,
$\bar{r} = \bar{r}(t_1, \bar{t}_1)$ and $\phi 
= \phi(t_1, \bar{t}_1)$ satisfy the following equations:
\begin{align}
\label{mPLR}
&\left\{
\begin{aligned}
 \frac{\partial q}{\partial \bar{t}_1}
 &= -2\bar{q}e^{2\phi},
\\
 \frac{\partial r}{\partial \bar{t}_1}
 &= 2\bar{r} e^{-2\phi},
\end{aligned}
\right.
\qquad
\left\{
\begin{aligned}
 \frac{\partial \phi}{\partial t_1}
 &= 2qr,
\\
 \frac{\partial \phi}{\partial \bar{t}_1}
 &= -2\bar{q}\bar{r},
\end{aligned}
\right.
\qquad
\left\{
\begin{aligned}
 \frac{\partial \bar{q}}{\partial t_1}
 &= -2qe^{-2\phi},
\\
 \frac{\partial \bar{r}}{\partial t_1}
 &= 2re^{2\phi}.
\end{aligned}
\right.
\end{align}
Originally,
the integrable Pohlmeyer-Lund-Regge system
was derived
in a study of the dynamics of relativistic vortices 
by Lund and Regge \cite{LR},
and 
independently 
in an investigation of the nonlinear
sigma model in field theory by Pohlmeyer \cite{Po}.
In \cite{JM3},
Jimbo and Miwa showed that the 
third Painlev\'e
equation is
obtained through a similarity reduction from the Pohlmeyer-Lund-Regge equation.
We will show later that the equations \eqref{mPLR} are
obtained by Miura transformation from the equations in \cite{JM3}.

This article is organized as follows. 
In Section 2, we formulate
the mPLR hierarchy 
based on the Sato theory 
by using the affine Lie algebra and group of type $A_1^{(1)}$.
Then
in Section 3,
we introduce
the B\"acklund transformations for 
the mPLR hierarchy 
and show that the group of these transformations 
provides
a realization of an extended affine Weyl group.
In Section 4,
we study a similarity reduction of the mPLR hierarchy
and give the action of the transformations defined
in Section 3 on the parameters of the similarity condition.
In Section 5 
the third Painlev\'e equation and its symmetry
are derived from the mPLR hierarchies.
In Section 6, 
we recall the similarity reduction to the fourth Painlev\'e
equation and
show that
the second Painlev\'e equation can 
also be derived from the mPLR hierarchy.

\section{Modified Pohlmeyer-Lund-Regge hierarchy}

The zero-curvature equation for the mPLR \eqref{zerocurv-mPLR} is
a compatibility condition for the linear equations
\begin{equation}
\label{linear-mPLR}
 \frac{\partial Y}{\partial t_1} = B_1Y,
 \qquad
 \frac{\partial Y}{\partial \bar{t}_1} = \bar{B}_1Y.
\end{equation}
In this section we construct 
formal solutions of these equations
based on the
theory of the affine Lie algebra $\hat{\mathfrak{sl}}_2$ and 
its group \cite{BtK1}, \cite{tau}.
This is equivalent to a modification of the 2-component 
Toda lattice hierarchy \cite{UT}.

\subsection{Gauss decomposition and $\tau$-functions}

Let $\mathfrak{g} = \mathfrak{sl}_2(\IC)$ and $\hat{\mathfrak{g}} 
= \mathfrak{sl}_2 \otimes \IC[\zeta, \zeta^{-1}] 
\oplus \IC c \oplus \IC d$ the
associated affine Lie algebra of 
type $A_1^{(1)}$.
$\hat{\mathfrak{g}}$ is generated 
by the Chevalley generators $h_i$, $e_i^\pm$ ($i=0,1$) and $d$ with
the defining relations 
\begin{align*}
&[h_i, h_j] = 0, 
\quad [h_i, e_j^\pm] = \pm a_{ij}e_j^\pm,
\quad [e_i^+, e_j^-] = \delta_{ij}h_i,
\quad
(\mathrm{ad}e_i^\pm)^{1-a_{ij}}e_j^\pm = 0,
\\
&[d, h_i] = 0,
\quad
[d, e_i^+] = \delta_{i,0}e_0^+,
\quad
[d, e_i^-] = -\delta_{i,0}e_0^-.
\end{align*}
Here $A = (a_{ij})_{i,j=0,1} 
= \begin{bmatrix} 2 & -2 \\ -2 & 2 \end{bmatrix}$ is 
the Cartan matrix of type $A_1^{(1)}$.
$\hat{\mathfrak{g}}$ has the following 
triangular decomposition:
\begin{equation}
\label{triangle}
 \mathfrak{\hat{g}} = 
 \mathfrak{\hat{n}}_- \oplus 
 \mathfrak{\hat{h}} \oplus
 \mathfrak{\hat{n}}_+,
\end{equation}
where $\mathfrak{\hat{n}}_\pm$ denotes the subalgebra of $\mathfrak{\hat{g}}$ generated 
by $e_i^\pm$ and 
$\mathfrak{\hat{h}} =  \IC h_0 \oplus \IC h_1 \oplus \IC d$ is 
the Cartan subalgebra of $\hat{\mathfrak{g}}$.
Note that $h_0+h_1 = c$.
The homogeneous Heisenberg subalgebra $\hat{\mathfrak{s}}$ of $\hat{\mathfrak{g}}$ is
given by
$\hat{\mathfrak{s}} = \oplus_{n\in \IZ} \IC H_n \oplus \IC c$ with
the commutation relations $[H_m, H_n]
= m\delta_{m+n,0}c$

Let $(\pi_j, L_j)$, $j=0,1$, be the basic representation with the highest 
weight vector $|v_j \rangle$.
The action of the Chevalley generators is given by
\begin{align}
\label{basic}
 \pi_j(e_i^+)|v_j \rangle = 0,
 \qquad
 \pi_j(h_i)|v_j \rangle = \delta_{ij} |v_j \rangle,
 \qquad
 \pi_j(e_i^-)^{\delta_{ij}+1}|v_j \rangle = 0.
\end{align}
We will denote by $(\pi_j^*, L_j^*)$ the dual representation of $(\pi_j, L_j)$, 
with the highest weight vector $\langle v_j |$ and the 
action of the Chevalley generators is given by
$$
 \pi_j^*(e_i^\pm) = \pi_j(e_i^\mp),
 \qquad
 \pi_j^*(h_i) = \pi_j(h_i),
 \qquad
 |v_i \rangle^* = \langle v_i|.
$$
It is known that
$L_j$ has the structure $L_j \simeq \IC[x_1, x_2, \dots] \otimes \IC[Q]$,
where $Q$ is the root lattice of $\mathfrak{sl}_2(\IC)$ and $\IC[Q]$ is the group 
algebra of $Q$,
and the action of $e_i^\pm$ is given in terms of vertex operators.
But we will not need this realization in this paper.

Let $(\pi, L) = (\pi_0 \oplus \pi_1, L_1 \oplus L_2)$. 
Since $(\pi, L)$ is an integrable representation, 
we can define the actions of $\exp\pi(e^\pm_i)$ $(i=0,1)$ on $L$.
Set
\begin{align}
\label{Weyl-generator}
 S_i :=& \exp\pi(e_i^-)\exp(-\pi(e_i^+))\exp\pi(e_i^-)
\qquad (i=0,1).
\end{align}
$S_0$ and $S_1$ generate the affine Weyl group of type $A_1^{(1)}$.

Now we define the $\tau$-functions of the mPLR hierarchy.
Let $\hat{G}$ be the affine Lie group generated
by $\exp s\pi(e_i^\pm)$, $i=0,1$, $s \in \IC$ and
put $g(0) \in \hat{G}$ as the initial value of the 
hierarchy.
We introduce the time-evolution of $g(0)$ with respect to time
variables $t =(t_1, t_2, \dots)$ and $\bar{t} = (\bar{t}_1, \bar{t}_2, \dots)$ by
\begin{equation}
\label{time-evolution}
g(t,\bar{t}) 
= \exp(\sum_n t_n H_n) g(0) \exp(-\sum_n \bar{t}_n H_{-n}).
\end{equation}
In the definition \eqref{time-evolution} there are infinitely many independent 
variables but in the following
we shall take almost all $t_i$ and $\bar{t}_i$'s to be zero.
We mention that
our definition \eqref{time-evolution} differs
from the previous paper \cite{KK1}.
Here we use not only the variable $t_i$ but also $\bar{t}_i$.

The $\tau$-functions of the mPLR hierarchy are defined by
\begin{equation}
\label{tau-mn}
 \tau_{m,n}^{(i)}(t,\bar{t})
:= \langle v_i| T^{-m} g(t,\bar{t}) T^n | v_i \rangle,
\qquad
(i=0,1,
\ 
m,n \in \IZ)
\end{equation}
where $T := S_0S_1$.

We will establish the relation between the $\tau$-function and the
components of $g$.
We require that the element $g = g(t,\bar{t})$ \eqref{time-evolution} can
be written as
\begin{equation}
\label{GaussDecomp}
g= g_{<0}^{-1}g_0g_{>0},
\end{equation}
where
$g_{<0} \in \exp\pi(\hat{\mathfrak{n}}_-)$,
$g_{>0} \in \exp\pi(\hat{\mathfrak{n}}_+)$ and $g_0 \in \exp\pi(\hat{\mathfrak{h}}')$,
$\hat{\mathfrak{h}}' := \IC h_0 \oplus \IC h_1$.
The factorization \eqref{GaussDecomp} is called the Gauss decomposition.
It is known that if an element $g \in \hat{G}$ belongs to the 
dense open subset of $\hat{G}$ called a big cell,
then $g$ is written as \eqref{GaussDecomp} \cite{Wil}.
Put 
\begin{align}
\label{gminus}
  g_{<0}(t,\bar{t})  &= \exp \left( q(t,\bar{t})e_0^- + r(t,\bar{t})e_1^-  
+ \cdots \right),
\\
\label{gzero}
g_0(t,\bar{t}) &= \exp (\phi_0(t,\bar{t})h_0 + \phi_1(t,\bar{t})h_1),
\\
\label{gplus}
  g_{>0}(t,\bar{t})  &= \exp\left( \bar{q}(t,\bar{t})e_1^+ + \bar{r}(t,\bar{t})e_0^+ 
+ \cdots \right).
\end{align}
Since $\exp\pi(\hat{\mathfrak{n}}_+)$ and $\exp
\pi(\hat{\mathfrak{n}}_-)$ stabilize $|v_i\rangle$ and $\langle v_i|$ ($i=0,1$) respectively,
by substituting \eqref{GaussDecomp} into \eqref{tau-mn} we have
\begin{equation}
\label{tauzero}
\begin{aligned}
\tau_{0,0}^{(i)}(t,\bar{t}) 
&= \langle v_i | g(t,\bar{t}) | v_i \rangle 
= \langle v_i | g_0(t,\bar{t}) | v_i \rangle 
\\
&= \langle v_i | \exp(\phi_0(t,\bar{t}) h_0 + \phi_1(t,\bar{t}) h_1) | v_i \rangle
= e^{\phi_i(t,\bar{t})}.
\end{aligned}
\end{equation}
Next we give the relation to 
the components of $g_{<0}$ \eqref{gminus}, 
$g_{>0}$ \eqref{gplus} and the $\tau$-functions.
By the relations \eqref{basic},
the  action of $g = g(t,\bar{t})$ on the vector $|v\rangle 
:= |v_0 \rangle + |v_1 \rangle$ is computed as follows:
\begin{align*}
 g |v\rangle
 = g_{<0}^{-1}g_0 |v \rangle 
           &= \tau_{00}^{(0)}(1 - qe_0^-  
                   + \cdots) |v_0 \rangle  
          + \tau_{00}^{(1)}(1 - re^-_1 
                   + \cdots) | v_1 \rangle.
\end{align*}
Similarly, by calculating $\langle v |g$,
we obtain
\begin{align*}
\begin{aligned}
 q &= -\frac{\langle v_0 | e_0^+ g |v_0 \rangle}{\tau_{00}^{(0)}},
\\
 r &= -\frac{\langle v_1 | e_1^+ g |v_1 \rangle}{\tau_{00}^{(1)}},
\end{aligned}
\qquad
\begin{aligned}
 \bar{q} &= \frac{\langle v_1 | g e_1^- |v_1 \rangle}{\tau_{00}^{(1)}},
\\
 \bar{r} &= \frac{\langle v_0 | g e_0^- |v_0 \rangle}{\tau_{00}^{(0)}}.
\end{aligned}
\end{align*}
On the other hand, we have the action of $T^m$ on $\langle v|$ and $|v \rangle$,
for example,
$T|v_0 \rangle
=  S_0S_1| v_0 \rangle
=  S_0| v_0 \rangle
=  e_0^- |v_0 \rangle$.
By comparing these results, we have
\begin{align}
\label{qrqbarrbar-tau}
 q &= -\frac{\tau_{1,0}^{(0)}}{\tau_{0,0}^{(0)}},
\qquad
 r = \frac{\tau_{-1,0}^{(1)}}{\tau_{0,0}^{(1)}},
\qquad
 \bar{q} = -\frac{\tau_{0,-1}^{(1)}}{\tau_{0,0}^{(1)}},
\qquad
 \bar{r} = \frac{\tau_{0,1}^{(0)}}{\tau_{0,0}^{(0)}}.
\end{align}

\subsection{Sato-Wilson equations and Zero-curvature equations}

From now on,
we consider the following level-$0$ realization of $\hat{\mathfrak{g}}$:
\begin{equation}
\label{realization}
\begin{aligned}
&e_0^+ \mapsto \begin{bmatrix}0 & 1 \\ 0 & 0 \end{bmatrix},
\quad
e_1^+ \mapsto \begin{bmatrix} 0 & 0  \\ \zeta  & 0 \end{bmatrix}, 
\\
&e_0^- \mapsto \begin{bmatrix} 0 & \zeta^{-1} \\ 0 & 0 \end{bmatrix}, \quad
e_1^- \mapsto \begin{bmatrix} 0 & 0 \\ 1 & 0 \end{bmatrix}, \quad
h_1 \mapsto \begin{bmatrix} 1 & 0 \\ 0 & -1 \end{bmatrix},
\end{aligned}
\end{equation}
$c \mapsto 0$ and $d \mapsto \zeta\frac{d}{d\zeta}$.
Then the generators of the homogeneous Heisenberg subalgebra are
realized by $H_n \mapsto 
\begin{bmatrix} \zeta^n & 0 \\ 0 & -\zeta^n \end{bmatrix}$.
The elements of the affine Lie group $\hat{G}$
\eqref{gminus}, \eqref{gzero} and \eqref{gplus} can be written 
by means of the following $2\times 2$ matrices:
\begin{align}
\label{defW}
 g_{<0}(t,\bar{t}) =& I + W_1(t,\bar{t})\Lambda^{-1} + W_2(t,\bar{t}) \Lambda^{-2}  + \cdots,
 \\
\label{defg0}
  g_0(t,\bar{t}) =& \exp \Phi(t,\bar{t}),
 \\
\label{defWbar}
 g_{>0}(t,\bar{t}) =& I + \bar{W}_1(t,\bar{t})\Lambda 
 + \bar{W}_2(t,\bar{t})\Lambda^2 + \cdots,
\end{align}
where $\Lambda = \begin{bmatrix} 0 & 1 \\ \zeta & 0 \end{bmatrix}$,
and $\Phi(t,\bar{t}), W_i(t,\bar{t}), 
\bar{W}_i(t,\bar{t})$ ($i=1,2,\dots$) are $2\times 2$ diagonal matrices.
Especially, we see that
\begin{equation}
\label{W1}
W_1(t,\bar{t}) = \begin{bmatrix} q & 0 \\ 0 &  r \end{bmatrix},
\quad
\Phi(t,\bar{t}) = \begin{bmatrix} \phi & 0 \\ 0 & -\phi \end{bmatrix}
\quad \mbox{and} \quad
\bar{W}_1(t,\bar{t}) = \begin{bmatrix} \bar{q} & 0 \\ 0 &  \bar{r} \end{bmatrix},
\end{equation}
where
\begin{equation}
 e^{\phi} = e^{\phi_1-\phi_0}
= \frac{\tau^{(1)}_{0,0}}{\tau^{(0)}_{0,0}}.
\end{equation}

By the defining relation \eqref{time-evolution},
$g=g(t,\bar{t})$ satisfies the equation
\begin{equation}
\label{derivative-g}
\frac{\partial g}{\partial t_n} = H_ng,
\qquad 
\frac{\partial g}{\partial \bar{t}_n} = -g H_{-n}.
\end{equation}
Using the Gauss decomposition \eqref{GaussDecomp},
$g_{<0} = g_{<0}(t,\bar{t})$ and $g_{\ge 0} := g_0(t,\bar{t})g_{>0}(t,\bar{t})$ satisfy
the following linear equations,
the so-called Sato-Wilson equations:
\begin{align}
\label{homogSW}
&\left\{
\begin{aligned} 
\frac{\partial g_{<0}}{\partial t_n} 
=& -\left(g_{<0} H_n g_{<0}^{-1}\right)_{<0}g_{<0},
\\
\frac{\partial g_{\ge 0} }{\partial t_n} 
=& \left(g_{<0} H_n g_{<0}^{-1}\right)_{\ge 0}g_{\ge 0},
\end{aligned}
\right.
\qquad
\left\{
\begin{aligned}
\frac{\partial g_{<0}}{\partial \bar{t}_n} 
=& \left(g_{\ge 0} H_{-n} g_{\ge 0}^{-1}\right)_{<0}g_{<0},
\\
\frac{\partial g_{\ge 0} }{\partial \bar{t}_n} 
=& -\left(g_{\ge 0} H_{-n} g_{\ge 0}^{-1}\right)_{\ge 0}g_{\ge 0},
\end{aligned}
\right.
\end{align}
where 
$(\;)_{<0}$ and $(\;)_{\ge 0}$ denote the negative and nonnegative power
parts of $\Lambda$ (not of $\zeta$),
i.e.,
for a matrix $A = \sum_{i \in \IZ} A_i \Lambda^i$,
we put $A_{<0} =  \sum_{i < 0} A_i \Lambda^i$ and
$A_{\ge 0} = \sum_{i \ge 0} A_i \Lambda^i$.
The equations \eqref{homogSW} for $n=1$ give
the
modified Pohlmeyer-Lund-Regge equations \eqref{mPLR}.
Also, from the formulas \eqref{derivative-g} we obtain the
relations
\begin{equation}
\label{w-tau}
\left\{
\begin{aligned}
 w_{21} &= -\frac{1}{2}\frac{\partial}{\partial t_1}\log\tau_{0,0}^{(0)},
\\
 w_{22} &= \frac{1}{2}\frac{\partial}{\partial t_1}\log\tau_{0,0}^{(1)},
\end{aligned}
\right.
\qquad
\quad
\left\{
\begin{aligned}
 \bar{w}_{21} &= -\frac{1}{2}\frac{\partial}{\partial \bar{t}_1}\log\tau_{0,0}^{(1)}, 
\\
 \bar{w}_{22} &= \frac{1}{2}\frac{\partial}{\partial \bar{t}_1}\log\tau_{0,0}^{(0)}.
\end{aligned}
\right.
\end{equation}
Here $W_2 := \mathrm{diag}(w_{21}, w_{22})$ and $\bar{W}_2 
:= \mathrm{diag}(\bar{w}_{21}, \bar{w}_{22})$.
Notice that $H_n = \Lambda^{2n}$.

Next 
we will construct formal solutions to the linear equation \eqref{linear-mPLR}.
Let $\alpha, \beta \in \IC$ be constants.
We define the wave 
functions $\Psi(\zeta;\alpha,t,\bar{t})$ and $\bar{\Psi}(\zeta;\beta,t,\bar{t})$ by
\begin{equation}
\label{BAfunction}
\begin{aligned}
 &\Psi(\zeta; \alpha, t, \bar{t})
 := g_{<0} \zeta^{\alpha H_0} \Psi_0(\zeta;t),
\\
 &\bar{\Psi}(\zeta; \beta, t, \bar{t})
  := g_{\ge 0}\zeta^{\beta H_0} \bar{\Psi}_0(\zeta;\bar{t}),
\end{aligned}
\end{equation}
where 
\begin{align}
 \Psi_0(\zeta; t)
&:= \exp \left( \sum_{n} t_n H_n \right)
= \begin{bmatrix} 
    e^{t_1\zeta + t_2\zeta^2 + \cdots}
    & 0 \\ 0 &
     e^{-(t_1\zeta + t_2\zeta^2 + \cdots)} 
   \end{bmatrix},
\label{Psi0}
\\
 \bar{\Psi}_0(\zeta; \bar{t})
&:= \exp \left( \sum_{n} \bar{t}_n H_{-n} \right)
= \begin{bmatrix} 
    e^{\bar{t}_1\zeta^{-1} + \bar{t}_2\zeta^{-2} + \cdots}
    & 0 \\ 0 &
     e^{-(\bar{t}_1\zeta^{-1} + \bar{t}_2\zeta^{-2} + \cdots)} 
   \end{bmatrix}.
\label{Psi0bar}
\end{align}
Obviously, $\Psi_0$ and $\bar{\Psi}_0$ satisfy the linear equations 
$$
\frac{\partial \Psi_0}{\partial t_n} = H_n \Psi_0,
\qquad
\frac{\partial \bar{\Psi}_0}{\partial \bar{t}_n} = H_{-n} \bar{\Psi}_0.
$$
By this relation and the Sato-Wilson equations \eqref{homogSW},
the functions
$Y = \Psi, \bar{\Psi}$ satisfy 
the linear equations 
\begin{equation}
\label{linearDE}
 \frac{\partial Y}{\partial t_n} = B_nY,
 \qquad
 \frac{\partial Y}{\partial \bar{t}_n} 
= \bar{B}_n Y,
\end{equation}
where
\begin{equation}
\label{defofBn}
 B_n := \left(g_{<0} H_n g_{<0}^{-1}\right)_{\ge 0},
 \qquad
 \bar{B}_n := \left(g_{\ge 0} H_{-n} g_{\ge 0}^{-1}\right)_{<0}.
\end{equation}
The lowest case $n=1$ gives $B_1$ in \eqref{B1} and $\bar{B}_1$ in \eqref{B1bar}.

Compatibility conditions for the linear equations
\eqref{linearDE} lead to the following zero-curvature equations:
\begin{equation}
\label{zerocurv}
\begin{aligned}
&\left[ \frac{\partial}{\partial t_n}
 -B_n, \frac{\partial}{\partial t_m} - B_m
\right] = 0,
\qquad
\left[ \frac{\partial}{\partial \bar{t}_n}
 - \bar{B}_n, 
\frac{\partial}{\partial \bar{t}_m} - \bar{B}_m  
\right] = 0,
\\
&\left[ \frac{\partial}{\partial t_n}
 -B_n, \frac{\partial}{\partial \bar{t}_m} - \bar{B}_m 
\right] = 0.
\end{aligned}
\end{equation}
Notice that
by the Sato-Wilson equations \eqref{homogSW}, 
the linear operators of \eqref{zerocurv} can be represented as follows:
\begin{equation}
\label{del-B}
\begin{aligned}
\frac{\partial}{\partial t_n} - B_n 
&= g_{<0} \cdot \left(\frac{\partial}{\partial t_n} - H_n\right) \cdot g_{<0}^{-1}
= g_{\ge 0} \cdot \frac{\partial}{\partial t_n} \cdot g_{\ge 0}^{-1},
\\
\frac{\partial}{\partial t_n} - \bar{B}_n 
&= g_{\ge 0} \cdot \left(\frac{\partial}{\partial \bar{t}_n} - H_{-n} \right)
\cdot g_{\ge 0}^{-1}
= g_{< 0} \cdot \frac{\partial}{\partial \bar{t}_n} \cdot g_{< 0}^{-1}.
\end{aligned}
\end{equation}
Here the product $\cdot$ means the composition of operators.

\subsection{Lax equations and a conserved density}

It seems that the equations of hierarchy \eqref{zerocurv}, \eqref{defofBn} contain
an infinite number of dependent variables (components of $W_i$ and $\bar{W}_i$,
$i=1,2,3,\dots$).
However, we will show that the coefficients of $B_n$, $\bar{B}_n$ are differential polynomial
of $(q, r)$, $(\bar{q}, \bar{r})$ with respect to $t_1$, $\bar{t}_1$ respectively.
Set
\begin{equation}
\label{homogeneous-Lax-operator}
\begin{aligned}
 L :=& g_{<0} H_0 g_{<0}^{-1} 
= H_0 + U_1\Lambda^{-1} + U_2\Lambda^{-2}
+ \cdots,
 \\
 \bar{L} :=& g_{\ge 0} H_0 g_{\ge 0}^{-1} 
= g_0\left( H_0 + \bar{U}_1\Lambda + \bar{U}_2 \Lambda^2 + \cdots \right)g_0^{-1},
\end{aligned}
\end{equation}
where $U_i = \begin{bmatrix} u_{i1} & 0 \\ 0 & u_{i2} \end{bmatrix}$,
$\bar{U}_i = \begin{bmatrix} \bar{u}_{i1} & 0 \\ 0 & \bar{u}_{i2} \end{bmatrix}$.
By the Sato-Wilson equations \eqref{homogSW} 
we have the following relations:
\begin{equation}
\label{defofU2n}
U_{2n} = \frac{\partial \Phi}{\partial t_n},
\qquad
\bar{U}_{2n} = -\frac{\partial \Phi}{\partial \bar{t}_n},
\end{equation}
and
\begin{equation}
\label{defofU2n+1}
U_{2n+1} = -\frac{\partial W_1}{\partial t_n},
\qquad
\bar{U}_{2n+1} = -\frac{\partial \bar{W}_1}{\partial \bar{t}_n}.
\end{equation}
The matrices $L$ and $\bar{L}$ satisfy the following system of equations (Lax equations):
\begin{equation}
\label{homogLaxeq}
\left\{
\begin{aligned}
\frac{\partial L}{\partial t_n}
 &= [B_n, L],
\\
\frac{\partial \bar{L}}{\partial t_n}
 &= [B_n, \bar{L}],
\end{aligned}
\right.
\qquad
\left\{
\begin{aligned}
\frac{\partial L}{\partial \bar{t}_n}
 =& [\bar{B}_n, L],
\\
\frac{\partial \bar{L}}{\partial \bar{t}_n}
 =& [\bar{B}_n, \bar{L}].
\end{aligned}
\right.
\end{equation}
\begin{proposition}
\label{prop1}
The components of $U_n$ \eqref{homogeneous-Lax-operator}
can be described by differential
polynomials of $q$ and $r$ with respect to $t_1$,
and $\bar{U}_n$ by differential
polynomials of $\bar{q}$ and $\bar{r}$ with respect to $\bar{t}_1$.
\end{proposition}
\begin{proof}
By the definition of $L$ \eqref{homogeneous-Lax-operator}, 
$L^2 = g_{<0} H_0^2 g_{<0}^{-1} = I$ and hence
\begin{equation}
\label{L^2=I}
 H_0U_j  + U_j\Lambda^{-j}H_0\Lambda^j
 + \sum_{k=1}^{j-1} U_k\Lambda^{-k} U_{j-k}\Lambda^k
= 0 
\qquad
 (j=1,2,\dots).
\end{equation}
On the other hand,
the derivation of $L$ with respect to $t_1$ in \eqref{homogLaxeq} is expressed by
\begin{align*}
\frac{\partial L}{\partial t_1}
=& [L, (\zeta L)_{<0}]
= [L_{\ge -2}, (\zeta L)_{<0}]
\\
=& [H_0 + U_1\Lambda^{-1} + U_{2}\Lambda^{-2},
U_3\Lambda^{-3} + U_4\Lambda^{-4} + \cdots].
\end{align*}
By comparing the coefficients
of $\Lambda^{2k}$ and $\Lambda^{2k+1}$ we have
\begin{equation}
\label{u2k+1}
\begin{aligned}
\frac{\partial U_{2k}}{\partial t_1} 
&= [U_1\Lambda, U_{2k+1}\Lambda^{-1}],
\\
\frac{\partial U_{2k+1}}{\partial t_1}\Lambda 
&=
[H_0, U_{2k+3}\Lambda] +
[U_1\Lambda, U_{2k+2}] +
[U_2, U_{2k+1}\Lambda]
\end{aligned}
\end{equation}
($k=1,2,3,\dots$).
Together with the relation \eqref{L^2=I},
the $U_j$ ($j \ge 2$) are recursively determined as
differential polynomials of $q$ and $r$ with respect to $t_1$.
$\bar{L}$ can be treated in the same way.
\end{proof}
We give some examples of $u_{n1}$ and $u_{n2}$ for small $n$.
\begin{align}
\label{u3qrrep}
&\left\{
\begin{aligned}
 u_{31} &= -q', 
\\
 u_{32} &= -r',
\end{aligned}
\right.
\qquad
\left\{
\begin{aligned}
 u_{41} & = q'r-qr'-2q^2r^2,
\\
 u_{42} &= -u_{41},
\end{aligned}
\right.
\qquad
\left\{
\begin{aligned}
 u_{51} &= -\frac{q''}{2} + 2q^2r'+4q^3r^2,
\\
 u_{52} &= \frac{r''}{2} + 2q'r^2 - 4q^2r^3,
\end{aligned}
\right.
\\
\label{u6qrrep}
&\left\{
\begin{aligned}
 u_{61} 
&= \frac{1}{2}(q''r+qr''-q'r')-4q^3r^3,
\\
 u_{62} &= -u_{61},
\end{aligned}
\right.
\qquad
\left\{
\begin{aligned}
 u_{71} 
&= -\frac{q'''}{4}+3qq'r' + 6q^2r^2q',
\\
 u_{72} 
&=-\frac{r'''}{4} + 3rq'r' + 6q^2r^2r',
\end{aligned}
\right.
\end{align}
where $f'$ means $\frac{\partial f}{\partial t_1}$.
By Proposition \ref{prop1},
the equations \eqref{defofU2n+1} can be regarded as
nonlinear differential equations for the variables $(q,r)$ or
$(\bar{q}, \bar{r})$.
For example,
we have
\begin{equation}
\label{derivativeNLS}
\left\{
\begin{aligned}
\frac{\partial q}{\partial t_2} &= \frac{q''}{2} -2q^2r' - 4q^3r^2,
\\
\frac{\partial r}{\partial t_2} &= -\frac{r''}{2} - 2r^2q' + 4q^2r^3,
\end{aligned}
\right.
\qquad
\left\{
\begin{aligned}
\frac{\partial q}{\partial t_3} &= \frac{q'''}{4} -3qq'r' - 6q^2r^2q',
\\
\frac{\partial r}{\partial t_3} &= \frac{r'''}{4} -3q'rr' - 6q^2r^2r'.
\end{aligned}
\right.
\end{equation}
The first couple of equations is the derivative nonlinear Schr\"odinger equation 
we have studied in \cite{KK1}, \cite{KK2} and the 
second one is a modification of the coupled modified KdV equation.

Other important functions 
are $W_2$ in \eqref{defW} and $\bar{W}_2$ in \eqref{defWbar}, which
are written as in \eqref{w-tau} by 
means of $\tau$-functions.
They are not differential polynomials of $q, r$ or $\bar{q}, \bar{r}$.
By the Sato-Wilson equations \eqref{homogSW}
they satisfy the relations
\begin{equation}
\label{W2tnbibun}
\left\{
\begin{aligned}
 &\frac{\partial W_2}{\partial t_n} = -U_{2n+2}
-U_{2n+1}(\Lambda W_1\Lambda^{-1}),
\\
 &\frac{\partial W_2}{\partial \bar{t}_n} = \bar{U}_{2n-2}
+ e^{2\Phi}\bar{U}_{2n-1}(\Lambda W_1 \Lambda^{-1})
\quad
(\bar{U}_0 := H_0),
\end{aligned}
\right.
\end{equation}
and
\begin{equation}
\label{W2bartnbibun}
\left\{
\begin{aligned}
 &\frac{\partial \bar{W}_2}{\partial t_n} = 
U_{2n-2} + e^{-2\Phi}
U_{2n-1}(\Lambda \bar{W}_1 \Lambda^{-1})
\quad (U_0 := H_0),
\\
 &\frac{\partial \bar{W}_2}{\partial \bar{t}_n} = 
-\bar{U}_{2n+2} -\bar{U}_{2n+1}(\Lambda \bar{W}_1 \Lambda^{-1}).
\end{aligned}
\right.
\end{equation}
So, the trivial identity $\partial^2 W_2/\partial t_n \partial t_m
= \partial^2 W_2/\partial t_m \partial t_n$ gives us
the $t_n$-conserved density $F_n = \partial W_2/\partial t_n$,
i.e.,
for all $m$ there exists a $t_n$-differential
polynomial $G_m$ such that
$\partial F_n/\partial t_m = \partial G_m/\partial t_n$.
For instance,
\begin{align}
\label{w21'}
&\left\{ \begin{aligned}
 \frac{\partial w_{21}}{\partial t_1}
&= 
-2q\hat{r},
\\
 \frac{\partial w_{22}}{\partial t_1} 
&= 
2\hat{q}r,
\end{aligned} \right.
\qquad
\left\{
\begin{aligned}
\frac{\partial w_{21}}{\partial t_2} 
&=  q\hat{r}' - q'\hat{r},
\\
\frac{\partial w_{22}}{\partial t_2} 
&= \hat{q}'r - \hat{q}r'.
\end{aligned}
\right.
\end{align}
Here we put
\begin{align}
\label{rhat}
\hat{q}
:=  \frac{1}{2}\frac{\partial q}{\partial t_1}-q^2r,
\qquad
\hat{r} 
:= -\frac{1}{2}\frac{\partial r}{\partial t_1} - qr^2.
\end{align}

\subsection{Miura-type transformation to the AKNS hierarchy}

Here
we describe the Miura-type transformations of the mPLR
hierarchy to the PLR hierarchy.
Miura-type transformations are defined by
\begin{align}
\label{Miu-}
g_{<0} &\mapsto \tilde{g}_{<0} := 
\begin{bmatrix} 1 & 0 \\ -r & 1 \end{bmatrix}g_{<0} 
= \cdots
+ \begin{bmatrix} 
                w_{21} & q \\ 
                \hat{r} & w_{22}-qr
        \end{bmatrix}\zeta^{-1} 
+ \begin{bmatrix}
                 1  &  0   \\ 
                 0 &  1
        \end{bmatrix},
\\
 g_0 &\mapsto
 \tilde{g}_0 := 
\begin{bmatrix} 1 & 0 \\ r & 1 \end{bmatrix}g_0 
\begin{bmatrix} 1 & \bar{q} \\ 0 & 1 \end{bmatrix}
=: \begin{bmatrix} a & b \\ c & d \end{bmatrix},
\\
\label{Miu+}
 g_{>0} &\mapsto \tilde{g}_{>0} 
:= \begin{bmatrix} 1 & -\bar{q} \\ 0 & 1 \end{bmatrix} g_{>0}
= \begin{bmatrix} 
          1 &  0  \\ 
          0 &  1
        \end{bmatrix} 
+ \begin{bmatrix}
              \bar{w}_{21}-\bar{q}\bar{r} & \hat{\bar{q}}  \\ 
              \bar{r} & \bar{w}_{22}
        \end{bmatrix}\zeta
+ \cdots .
\end{align}
where $\hat{r}$ is defined by \eqref{rhat} and 
we also put
\begin{align}
\label{qbarhat}
\hat{\bar{q}} 
:= \frac{1}{2}\frac{\partial \bar{q}}{\partial \bar{t}_1} - \bar{q}^2\bar{r},
\qquad
\hat{\bar{r}} 
:= -\frac{1}{2}\frac{\partial \bar{r}}{\partial \bar{t}_1}-\bar{q}\bar{r}^2.
\end{align}
We remark that the expressions of $\hat{r}$ in \eqref{Miu-} and $\hat{\bar{q}}$ 
in \eqref{Miu+} are proved 
by using the relations \eqref{homogeneous-Lax-operator} and \eqref{u3qrrep}.
The variables $(a,b,c,d, q, \hat{r}, \hat{\bar{q}}, \bar{r})$ obtained as above, 
satisfy the
following relations:
\begin{align}
 &\frac{\partial}{\partial t_1} 
 \begin{bmatrix} a & b \\ c & d \end{bmatrix}
 =\begin{bmatrix} 0 & -2q \\ 2\hat{r} & 0 \end{bmatrix}
 \begin{bmatrix} a & b \\ c & d \end{bmatrix},
\qquad
 \frac{\partial}{\partial \bar{t}_1} 
 \begin{bmatrix} a & b \\ c & d \end{bmatrix}
 = \begin{bmatrix} a & b \\ c & d \end{bmatrix}
\begin{bmatrix} 0 & 2\hat{\bar{q}} \\ -2\bar{r} & 0 \end{bmatrix}
\\
 &\frac{\partial q}{\partial \bar{t}_1} = -2ab,
 \quad
 \frac{\partial \hat{r}}{\partial \bar{t}_1} = 2cd,
 \quad
 \frac{\partial \hat{\bar{q}}}{\partial t_1} = 2bd,
 \quad
 \frac{\partial \bar{r}}{\partial t_1} = -2ac.
\label{PLR-qrqbarrbar}
\end{align}
This system of equations is the Pohlmeyer-Lund-Regge equation
discussed in \cite{JM3}.
Furthermore, we have
\begin{equation}
\label{NLS}
\left\{
\begin{aligned}
   \frac{\partial q}{\partial t_2} 
&= \frac{1}{2}\frac{\partial^2 q}{\partial t_1^2} + 4q^2 \hat{r},
\\
   \frac{\partial \hat{r}}{\partial t_2}
&= -\frac{1}{2}\frac{\partial^2 \hat{r}}{\partial t_1^2} - 4q \hat{r}^2,
\end{aligned}
\right.
\qquad
\quad
\left\{
\begin{aligned}
   \frac{\partial q}{\partial t_3} 
&= \frac{1}{4}\frac{\partial^3 q}{\partial t_1^3} 
+ 6q\hat{r}\frac{\partial q}{\partial t_1},
\\
   \frac{\partial \hat{r}}{\partial t_3}
&= \frac{1}{4}\frac{\partial^3 \hat{r}}{\partial t_1^3} + 6q \hat{r}\hat{r}'.
\end{aligned}
\right.
\end{equation}
They are the nonlinear Schr\"odinger equation
and the coupled modified KdV equation.

The Miura-type transformation 
of the wave functions \eqref{BAfunction},
i.e.,
$Y = \tilde{g}_{<0} \zeta^{\alpha H_0}\Psi_0$ or
$\tilde{g}_{\ge 0}\zeta^{\beta H_0}\bar{\Psi}_0$ satisfy the
linear equations
\begin{equation}
\label{linearDE-m}
 \frac{\partial Y}{\partial t_n} = \tilde{B}_nY,
 \qquad
 \frac{\partial Y}{\partial \bar{t}_n} 
= \tilde{\bar{B}}_n Y,
\end{equation}
where $\tilde{B}_n = (\tilde{g}_{<0}H_n\tilde{g}_{<0})^{-1}$,
$\tilde{\bar{B}}_n = (\tilde{g}_0\tilde{g}_{> 0}
H_{-n}\tilde{g}_{> 0}^{-1}\tilde{g}_0^{-1})_{<0}$.
For example,
\begin{align*}
 \tilde{B}_1 &= \begin{bmatrix}  0 & -2q \\ 2\hat{r} & 0 \end{bmatrix}
+ \begin{bmatrix}  1 & 0 \\ 0 & -1 \end{bmatrix}\zeta,
 \\
 \tilde{\bar{B}}_1 
&= \begin{bmatrix} a & b \\ c & d \end{bmatrix}
 \begin{bmatrix}  \zeta^{-1} & 0 \\ 0 & -\zeta^{-1} \end{bmatrix}
 \begin{bmatrix} a & b \\ c & d \end{bmatrix}^{-1}.
\end{align*}
Compatibility conditions for these equations give
the AKNS hierarchy.

\section{Action of the extended affine Weyl groups of type $A_1^{(1)}$}

\subsection{Basic transformations}

We now discuss symmetries in the mPLR hierarchy. 
In this section we construct six ``basic'' transformations 
on the variables $g_{<0}$, $g_0$ and $g_{>0}$.
First, we define
\begin{equation}
\label{sigmaaction}
 \sigma :  (t, \bar{t}, W_i(t,\bar{t}), \bar{W}_i(t,\bar{t}), \Phi(t,\bar{t})) 
\mapsto (\bar{t}, t, \bar{W}_i(\bar{t}, t), W_i(\bar{t}, t), -\Phi(\bar{t}, t)),
\end{equation}
($i = 1,2,\dots$).
The correspondence of  the coefficients is the following:
\begin{equation}
\label{defofsigma}
\sigma(\phi)(t,\bar{t}) = -\phi(\bar{t},t),
\qquad
\left\{
\begin{aligned}
\sigma(q)(t,\bar{t}) &= \bar{q}(\bar{t},t),
\\
\sigma(r)(t,\bar{t}) &= \bar{r}(\bar{t},t),
\end{aligned}
\right.
 \qquad 
\left\{
\begin{aligned}
 \sigma(\bar{q})(t,\bar{t}) &= q(\bar{t},t),
 \\
 \sigma(\bar{r})(t,\bar{t}) &= r(\bar{t},t).
\end{aligned}
\right.
\end{equation}
Second, 
we consider the birational transformations for the variables $\phi,
q, r, \bar{q}$ and $\bar{r}$.
In ref.~\cite{KK1}, 
we have constructed two types of extended affine Weyl group symmetries 
(the ``left action'' and the ``right action'')
for the derivative NLS hierarchy.
Here we extend these to the mPLR hierarchy.
We denote the matrices $S_0$, $S_1$ by
\begin{align}
\label{Weyls0s1}
 S_0   := \begin{bmatrix} 
                0  & \zeta^{-1}\\ 
               -\zeta &  0
        \end{bmatrix}
\qquad
 S_1   := \begin{bmatrix}
                0  & -1   \\ 
                1 &  0 
        \end{bmatrix},
\end{align}
They are the level-zero realization of  the elements \eqref{Weyl-generator}.
Note that $S_0$ and $S_1$ \eqref{Weyls0s1} satisfy the relation $S_0^4 = S_1^4  = I$.
Define the ``left action'' of operators $S_0$, $S_1$  for $g_{>0}, 
g_0$ and $g_{<0}$ as follows:
\begin{align}
\label{leftaction}
\begin{aligned}
 &s_i^{\mathrm{L}}(g_{<0})(t,\bar{t}) 
= G_i^{\mathrm{L}}(-t,\bar{t})g_{<0}(-t, \bar{t})S_i,
\\
 &s_i^{\mathrm{L}}(g_0)(t,\bar{t}) 
= G_i^{\mathrm{L}}(-t,\bar{t})g_0(-t,\bar{t})
\bar{G}_i^{\mathrm{L}}(-t,\bar{t})^{-1},
\\
 &s_i^{\mathrm{L}}(g_{>0})(t,\bar{t}) 
= \bar{G}_i^{\mathrm{L}}(-t,\bar{t})g_{>0}(-t, \bar{t}),
\end{aligned}
\end{align}
($i=0,1$)
where
\begin{equation}
\label{Finfty}
\begin{aligned}
 &G_0^{\mathrm{L}}  
     := \begin{bmatrix}
               -1/q &  0  \\ 
                 \zeta  & -q 
        \end{bmatrix},
\quad
 G_1^{\mathrm{L}}
  :=         \begin{bmatrix}
               -r &  1  \\ 
                 0 & -1/r
        \end{bmatrix},
\\
 &\bar{G}_0^{\mathrm{L}} := \begin{bmatrix}
           1 & 0 \\ -\zeta e^{2\phi}/q & 1
       \end{bmatrix},
\quad
  \bar{G}_1^{\mathrm{L}} :=   \begin{bmatrix}
               1 & 1/re^{2\phi}  \\ 
                 0 & 1 
        \end{bmatrix}.
\end{aligned}
\end{equation} 
We also define the ``right action''  $s_0^{\mathrm{R}}, s_1^{\mathrm{R}}$ by 
\begin{align}
\label{rightaction}
\begin{aligned}
 &s_i^{\mathrm{R}}(g_{<0})(t,\bar{t}) 
=  G_i^{\mathrm{R}}(t,-\bar{t})g_{<0}(t, -\bar{t}),
\\
 &s_i^{\mathrm{R}}(g_0)(t,\bar{t}) 
= G_i^{\mathrm{R}}(t,-\bar{t})
g_0(t,-\bar{t}) \bar{G}_i^{\mathrm{R}}(t,-\bar{t})^{-1},
\\ 
 &s_i^{\mathrm{R}}(g_{>0})(t,\bar{t}) 
= \bar{G}_i^{\mathrm{R}}(t,-\bar{t})g_{>0}(t, -\bar{t})S_i,
\end{aligned}
\end{align}
($i=0,1$) where
\begin{equation}
\label{GzeroGone}
\begin{aligned}
 &G_0^{\mathrm{R}}   
=  \begin{bmatrix}
                  1  &  -\zeta^{-1}e^{2\phi}/\bar{r}  \\
                  0  &      1   
      \end{bmatrix},
\quad
 G_1^{\mathrm{R}} = \begin{bmatrix}
            1     &  0  \\ 
        -1/\bar{q}e^{2\phi} &  1
  \end{bmatrix},
\\
 &\bar{G}_0^{\mathrm{R}}  = \begin{bmatrix}
               \bar{r} &  -\zeta^{-1}  \\ 
                 0  & 1/\bar{r} 
        \end{bmatrix},
\quad
 \bar{G}_1^{\mathrm{R}}  =  \begin{bmatrix}
               1/\bar{q} &  0  \\ 
                 -1 &  \bar{q} 
        \end{bmatrix}.
\end{aligned}
\end{equation}
For example, we have 
\begin{alignat}{2}
\label{szeroLe}
s_0^\mathrm{L}
(e^{\phi})(t,\bar{t}) 
&= -\frac{e^{\phi(-t,\bar{t})}}{q(-t,\bar{t})},
&\qquad
s_1^\mathrm{L}
(e^{\phi})(t,\bar{t}) 
&= -e^{\phi(-t,\bar{t})}r(-t,\bar{t}),
\\
\label{szeroRe}
s_0^\mathrm{R}
(e^{\phi})(t,\bar{t}) 
&= \frac{e^{\phi(t,-\bar{t})}}{\bar{r}(t,-\bar{t})},
&\qquad
s_1^\mathrm{R}
(e^{\phi})(t,\bar{t}) 
&= e^{\phi(t,-\bar{t})}\bar{q}(t,-\bar{t}).
\end{alignat}

Finally,
we construct the Dynkin diagram automorphism $\pi$ by
\begin{equation}
\label{piaction}
\begin{aligned}
 &\pi(g_{<0})(t,\bar{t}) = P^{-1}g_{<0}(-t, -\bar{t})P,
\\
 &\pi(g_0)(t,\bar{t}) = P^{-1}g_0(-t,-\bar{t})P,
\\
 &\pi(g_{>0})(t,\bar{t}) = P^{-1}g_{>0}(-t, -\bar{t})P,
\end{aligned}
\end{equation}
where $P = \begin{bmatrix} 0 & \zeta^{-1/2} \\ -\zeta^{1/2} & 0 \end{bmatrix}$.
This action was also introduced in ref.~\cite{KK1}.
The relation of the independent variables are
$\pi(e^{\phi(t,\bar{t})}) = e^{-\phi(-t,-\bar{t})}$ and
\begin{equation}
\label{defofpi}
\left\{
\begin{aligned}
\pi(q)(t,\bar{t}) &= -r(-t,-\bar{t}),
\\
\pi(r)(t,\bar{t}) &= -q(-t,-\bar{t}),
\end{aligned}
\right.
 \qquad 
\left\{
\begin{aligned}
 \pi(\bar{q})(t,\bar{t}) &= -\bar{r}(-t,-\bar{t}),
 \\
 \pi(\bar{r})(t,\bar{t}) &= -\bar{q}(-t,-\bar{t}).
\end{aligned}
\right.
\end{equation}
\begin{theorem}
\label{thm1}
The transformations 
\eqref{sigmaaction},
\eqref{leftaction}, 
\eqref{rightaction} and \eqref{piaction} commute with 
the Sato-Wilson equations \eqref{homogSW}.
For the variables $q, r, \bar{q}$ and $\bar{r}$,
they satisfy the relations
\begin{equation}
\begin{aligned}
 &\sigma^2 = (s_0^\mathrm{L})^2 = (s_1^\mathrm{L})^2 =
 (s_0^\mathrm{R})^2 = (s_1^\mathrm{R})^2 =
 \pi^2 = \mathrm{Id}, \quad
\\
  &\pi s_0^\mathrm{L} = s_1^\mathrm{L} \pi, \quad
  \pi s_0^\mathrm{R} = s_1^\mathrm{R} \pi, \quad
\\
 &\sigma s_0^\mathrm{L} = s_1^\mathrm{R} \sigma,
 \quad
 \sigma s_1^\mathrm{L} = s_0^\mathrm{R} \sigma,
\end{aligned}
\label{relation-qr}
\end{equation}
and for the variable $\phi$,
\begin{equation}
\label{relation-phi}
\begin{aligned}
 &\sigma^2 = (s_0^\mathrm{L})^4 = (s_1^\mathrm{L})^4 =
 (s_0^\mathrm{R})^4 = (s_1^\mathrm{R})^4 =
 \pi^2 = \mathrm{Id}, \quad
\\
  &\pi s_0^\mathrm{L} = -s_1^\mathrm{L} \pi, \quad
  \pi s_0^{\mathrm{R}} = -s_1^{\mathrm{R}} \pi,
\\
 &\sigma s_0^\mathrm{L} = -s_1^\mathrm{R} \sigma,
 \quad
 \sigma s_1^\mathrm{L} = -s_0^\mathrm{R} \sigma.
\end{aligned}
\end{equation}
Especially, the group $\widetilde{W}_\mathrm{L} 
= \langle s_0^\mathrm{L},
s_1^\mathrm{L}, \pi \rangle$ and $\widetilde{W}_\mathrm{R}
= \langle s_0^\mathrm{R}, s_1^\mathrm{R}, \pi \rangle$ generated
by these transformations on $q, r, \bar{q}, \bar{r}$ are 
extended affine Weyl groups of type $A_1^{(1)}$.
\end{theorem}
\begin{proof}
This is checked by a direct calculation.
For example,
by the relation $S_i^2 = -I$, we have
\begin{align*}
 &(s_i^\mathrm{L})^2(g_{<0})(t,\bar{t}) = (-I)g_{<0}(t,\bar{t})(-I) = g_{<0}(t,\bar{t}),
\\
 &(s_i^\mathrm{L})^2(g_{0})(t,\bar{t}) = (-I)g_{0}(t,\bar{t}),
\end{align*}
and $(s_i^\mathrm{L})^2(g_{>0})(t,\bar{t}) = g_{0}(t,\bar{t})$.
\end{proof}
We remark that the definition of transformations \eqref{leftaction}, 
\eqref{rightaction} is based on the Gauss decomposition of 
the loop group,
which is discussed in \cite{Nhon}, \cite{KK1}.

\subsection{Relations among the $\tau$-functions on the $A_1^{(1)}$-root lattice}

By Theorem \ref{thm1},
we automatically obtain discrete integrable systems
by composing the basic transformations given in previous section.
First we consider the
compositions $T_\mathrm{L} := s^\mathrm{L}_0s^\mathrm{L}_1$ and
$T_\mathrm{R} := s^\mathrm{R}_0s^\mathrm{R}_1$.
They correspond to parallel transformations on 
the $A_1^{(1)}$-root lattice.

The functions $\tau_{m,n}^{(i)}(t,\bar{t})$ ($i=0,1$, $m,n \in \IZ$) defined
by \eqref{tau-mn} are realized as
\begin{equation}
\label{tau-mn-real}
 \tau_{m,n}^{(i)}(t,\bar{t}) = T_\mathrm{L}^mT_\mathrm{R}^n\tau_{0,0}^{(i)}(t,\bar{t}).
\end{equation}
Because the transformations $T_\mathrm{L}^k$ and $T_\mathrm{R}^k$ act
on the left and right index respectively,
e.g.
$$
 T_\mathrm{L}^k(\tau_{m,n}^{(i)}) = \tau_{m+k,n}^{(i)},
\qquad
 T_\mathrm{R}^k(\tau_{m,n}^{(i)}) = \tau_{m,n+k}^{(i)},
$$ 
we obtain a sequence of $\tau$-functions.
\begin{equation}
    \xymatrix{ 
\tau_{-1,1}^{(i)} \ar[rr]_{T_{\mathrm{L}}} &  & 
\tau_{0,1}^{(i)} \ar[rr]_{T_{\mathrm{L}}} & & 
\tau_{1,1}^{(i)} 
\\
& &  &  & 
\\ 
\tau_{-1,0}^{(i)} \ar[rr]_{T_{\mathrm{L}}}  \ar[uu]_{T_\mathrm{R}}&  & 
\tau_{0,0}^{(i)}  \ar[rr]_{T_\mathrm{L}}  \ar[uu]_{T_\mathrm{R}}
& & \tau_{1,0}^{(i)}  \ar[uu]_{T_\mathrm{R}}\\ 
& &  & &  
\\ 
\tau_{-1,-1}^{(i)} \ar[rr]_{T_{\mathrm{L}}}  \ar[uu]_{T_\mathrm{R}}& & 
\tau_{0,-1}^{(i)} \ar[rr]_{T_{\mathrm{L}}} \ar[uu]_{T_\mathrm{R}} &  & 
\tau_{1,-1}^{(i)}  \ar[uu]_{T_\mathrm{R}}
    }
\end{equation}

\begin{proposition}
The actions of the six basic transformations $\sigma, 
s_0^\mathrm{L}, s_1^\mathrm{L}, s_0^\mathrm{R}, s_1^\mathrm{R}$ and $\pi$ on
the $\tau$-functions $\tau_{m,n}$ $(m,n \in \IZ)$ are given by the
following table:
\begin{equation}
\label{Weylaction-tau}
\begin{array}{c|cc}
 & \tau_{m,n}^{(0)}(t,\bar{t}) & \tau_{m,n}^{(1)}(t,\bar{t})
\\
\hline
\sigma &
 \tau_{-n,-m}^{(1)}(\bar{t}, t)
& \tau_{-n,-m}^{(0)}(\bar{t}, t)
\\
 s_0^\mathrm{L}
& (-1)^m\tau^{(0)}_{1-m,n}(-t, \bar{t}) 
& (-1)^m\tau^{(1)}_{-m,n}(-t, \bar{t}) 
\\
 s_1^\mathrm{L}
& (-1)^m\tau^{(0)}_{-m,n}(-t, \bar{t}) 
& (-1)^{m+1}\tau^{(1)}_{-m-1,n}(-t, \bar{t}) 
\\
 s_0^\mathrm{R}
& (-1)^n\tau^{(0)}_{m,1-n}(t, -\bar{t}) 
& (-1)^n\tau^{(1)}_{m,-n}(t, -\bar{t}) 
\\
 s_1^\mathrm{R}
& (-1)^n\tau^{(0)}_{m,-n}(t, -\bar{t}) 
& (-1)^{n+1}\tau^{(1)}_{m,-n-1}(t, -\bar{t}) 
\\
\pi
& \tau_{-m,-n}^{(1)}(-t, -\bar{t})
& \tau_{-m,-n}^{(0)}(-t, -\bar{t})
\end{array}
\end{equation}

\end{proposition}
\begin{proof}
We prove only the case of $s_1^\mathrm{R}$.
The corresponding equations for the other elements are obtained similarly.
When $n\ge 0$,
the action of $s_1^\mathrm{R}$ on 
$\exp(\phi_{m,n}) 
:= T_m^\mathrm{L}T_n^\mathrm{R}\exp(\phi_{0,0})
= \tau_{m,n}^{(1)}/\tau_{m,n}^{(0)}$ is
\begin{align*}
 s_1^\mathrm{R}(\exp(\phi_{m,n}))
&= s_1^\mathrm{R}\left(s_0^\mathrm{R}s_1^\mathrm{R}\right)^n(\exp(\phi_{m,0}))
\\
&= \left(s_1^\mathrm{R}s_0^\mathrm{R}\right)^n s_1^\mathrm{R}(\exp(\phi_{m,0}))
= s_1^\mathrm{R}(\exp(\phi_{m,-n})).
\end{align*}
Due to the relation \eqref{szeroRe} we have
$$
s_1^\mathrm{R}(\exp(\phi_{m,-n}))
= (\tau^{(1)}_{m,-n}/\tau_{m,-n}^{(0)})
\times(-\tau^{(1)}_{m,-n-1}/\tau_{m,-n}^{(1)})
= -\tau^{(1)}_{m,-n-1}/\tau_{m,-n}^{(0)}.
$$ 
For $n< 0$,
the action of $s_1^\mathrm{R}$ on $\exp(\phi_{m,n})$ is given by
\begin{align*}
 s_1^\mathrm{R}(\exp(\phi_{m,n}))
&= s_1^\mathrm{R}\left(s_1^\mathrm{R}s_0^\mathrm{R}\right)^{-n}
(\exp(\phi_{m,0}))
\\
&= -\left(s_0^\mathrm{R}s_1^\mathrm{R}\right)^{-n-1} s_0^\mathrm{R}
(\exp(\phi_{m,0}))
= -s_0^\mathrm{R}(\exp(\phi_{m,-n-1})).
\end{align*}
By the relation \eqref{szeroRe},
we obtain $-s_0^\mathrm{R}(\exp(\phi_{m,-n-1}))
= -\tau^{(1)}_{m,-n-1}/\tau_{m,-n}^{(0)}$.
The signs of the transformed $\tau$-function is 
determined uniquely
by supposing $s_1^\mathrm{R}(\tau_{m,0}^{(0)})
= \tau_{m,0}^{(0)}$.
\end{proof}

Now we represent the variables \eqref{rhat} and \eqref{qbarhat}, 
obtained by 
the Miura transformations, in terms of the $\tau$-functions.
By definition, 
the actions of $T_\mathrm{L}$ and $T_\mathrm{R}$ on $g_{<0}$ \eqref{defW} are
\begin{alignat}{2}
 T_{\mathrm{L}}(g_{<0})(t,\bar{t}) 
&=  R_{\mathrm{L}}(t,\bar{t})g_{<0}(t, \bar{t}) S_0S_1,
&\qquad
 &R_\mathrm{L} := \begin{bmatrix} -\hat{q}/q & -q \\ 0 & -q/\hat{q} \end{bmatrix}
+\begin{bmatrix} 1 & 0 \\ 1/\hat{q} & 0 \end{bmatrix}\zeta,
\\
 T_{\mathrm{R}}(g_{<0})(t,\bar{t}) 
&=  R_{\mathrm{R}}(t,\bar{t})g_{<0}(t, \bar{t}),
&\qquad
 &R_\mathrm{R} 
= \begin{bmatrix} 0 & -e^{2\phi}/\bar{r} \\
   0 & -\bar{r}/\hat{\bar{r}}
\end{bmatrix}\zeta^{-1}
+\begin{bmatrix} 
 1 & 0 \\
\bar{r}^2/e^{2\phi}\hat{\bar{r}} & 1
\end{bmatrix}.
\end{alignat}
Notice that in these transformations
the signs of the dependent variables $t$ and $\bar{t}$ do 
not change.
So, by the relations \eqref{del-B}
the action $T_\mathrm{L}$  can be realized as
\begin{equation}
\begin{aligned}
 \frac{\partial}{\partial t_n} - T_\mathrm{L}(B_n) 
=& R_\mathrm{L} \cdot
\left( \frac{\partial}{\partial t_n} - B_n \right) 
\cdot R_{\mathrm{L}}^{-1}
\\
=& \frac{\partial}{\partial t_n}
- R_\mathrm{L} B_n  R_\mathrm{L}^{-1} 
- \frac{\partial R_\mathrm{L}}{\partial t_n}R_\mathrm{L}^{-1} 
\end{aligned}
\label{SchlesingerB}
\end{equation}
and by replacing $R_\mathrm{L}$ to $R_{\mathrm{R}}$, we get $T_\mathrm{R}$.
For example,
the left action is obtained by
\begin{equation*}
 T_\mathrm{L}(q) = \frac{q''}{4} - \frac{(q')^2}{4q} -\frac{q^2r'}{2} - q^3r^2,
\qquad
 T_\mathrm{L}(r) = \frac{1}{\hat{q}}.
\end{equation*}
These relations provide the following formulas:
\begin{align}
\label{Miura-tau}
 \hat{q} &= \frac{\tau_{1,0}^{(1)}}{\tau_{0,0}^{(1)}},
\qquad
 \hat{r} = -\frac{\tau_{-1,0}^{(0)}}{\tau_{0,0}^{(0)}},
\qquad
 \hat{\bar{q}} = \frac{\tau_{0,-1}^{(0)}}{\tau_{0,0}^{(0)}},
\qquad
 \hat{\bar{r}} = -\frac{\tau_{0,1}^{(1)}}{\tau_{0,0}^{(1)}}.
\end{align}

We now deduce the 
differential equations 
that are satisfied by the $\tau$-functions $\tau_{m,n}^{(i)}$.
By using the correspondence 
\eqref{qrqbarrbar-tau},
\eqref{w-tau}, \eqref{Miura-tau} and the 
action of $T^m_\mathrm{L}$, $T^n_\mathrm{R}$,
the equations \eqref{w21'} can be rewritten as
\begin{equation}
\begin{aligned}
D_{t_1}^2 \tau^{(i)}_{m,n}\cdot \tau^{(i)}_{m,n} 
&= 8\tau^{(i)}_{m+1,n} \tau^{(i)}_{m-1,n}, 
\\
D_{t_1}D_{t_2} \tau^{(i)}_{m,n}\cdot \tau^{(i)}_{m,n} 
&= 4D_{t_1}\tau^{(i)}_{m+1,n} \cdot \tau^{(i)}_{m-1,n}.
\end{aligned}
\end{equation}
Here $D_x$ is the Hirota bilinear operator defined
for a pair of functions $f$, $g$ and for an arbitrary polynomial $Q$ as follows:
$$
Q(D_x)f(x) \cdot g(x)  = Q(\frac{\partial}{\partial y})
(f(x + y)g(x - y))|_{y=0}.
$$
Furthermore, we obtain
\begin{align}
&(D_{t_1}^2 + 2D_{t_2}) \tau^{(i)}_{m,n} \cdot \tau^{(i)}_{m+1,n} = 0,
\\
&(D_{t_1}^3 - 4D_{t_3}) \tau^{(i)}_{m,n} \cdot \tau^{(i)}_{m+1,n} = 0
\end{align}
from the equations \eqref{defofU2n+1}.
Thanks to the symmetry $\sigma$ \eqref{sigmaaction},
we have similar equations with respect to $\bar{t}$ and the 
right index $n$.
For example, we have
\begin{equation}
D_{\bar{t}_1}^2 \tau^{(i)}_{m,n}\cdot \tau^{(i)}_{m,n} 
= 8\tau^{(i)}_{m,n+1} \tau^{(i)}_{m,n-1}.
\end{equation}

\subsection{Relations among the $\tau$-functions on the $A_1^{(1)}$-weight lattice}

Next we consider the compositions $\mathcal{T}_1 
:= s_0^\mathrm{R}s_0^\mathrm{L}\pi$ and $\mathcal{T}_2 
:= s_1^\mathrm{R}s_0^\mathrm{L}\pi$.
They correspond to shift operators on the $A_1^{(1)}$-weight lattice \cite{Nhon}.
By definition,
\begin{align}
&\left\{
\begin{aligned}
&\mathcal{T}_1(\tau^{(0)}_{m,n})(t,\bar{t}) 
=(-1)^{m+n}\tau^{(1)}_{m,n}(t,\bar{t}),
\\
&\mathcal{T}_1(\tau^{(1)}_{m,n})(t,\bar{t}) 
=(-1)^{m+n}\tau^{(0)}_{m+1,n+1}(t,\bar{t}),
\end{aligned}
\right.
\\
&\left\{
\begin{aligned}
&\mathcal{T}_2(\tau^{(0)}_{m,n})(t,\bar{t}) 
=(-1)^{m+n-1}\tau^{(1)}_{m,n-1}(t,\bar{t}),
\\
&\mathcal{T}_2(\tau^{(1)}_{m,n})(t,\bar{t}) 
=(-1)^{m+n}\tau^{(0)}_{m+1,n}(t,\bar{t})
\end{aligned}
\right.
\end{align}
hold.
So we arrange the family of $\tau$-functions on the
following diagram:
\begin{equation}
\label{weight-lattice}
    \xymatrix{ 
\tau_{-1,1}^{(0)} \ar[rd]_{\mathcal{T}_2} &  & 
\tau_{0,1}^{(0)} \ar[rd]_{\mathcal{T}_2} & & 
\tau_{1,1}^{(0)} \\
& \tau_{-1,0}^{(1)} \ar[rd]_{\mathcal{T}_2} \ar[ru]_{\mathcal{T}_1} &  & 
\tau_{0,0}^{(1)} \ar[rd]_{\mathcal{T}_2} \ar[ru]_{\mathcal{T}_1} & 
\\ 
\tau_{-1,0}^{(0)} \ar[rd]_{\mathcal{T}_2} \ar[ru]_{\mathcal{T}_1} &  & 
\tau_{0,0}^{(0)} \ar[rd]_{\mathcal{T}_2}  
\ar[ru]_{\mathcal{T}_1}  \ar[rr]_{T_\mathrm{L}}  \ar[uu]_{T_\mathrm{R}}
& & 
\tau_{1,0}^{(0)} 
\\ 
& \tau_{-1,-1}^{(1)} \ar[rd]_{\mathcal{T}_2} \ar[ru]_{\mathcal{T}_1}  &  & 
\tau_{0,-1}^{(1)} \ar[rd]_{\mathcal{T}_2} \ar[ru]_{\mathcal{T}_1}  &  
\\ 
\tau_{-1,-1}^{(0)} \ar[ru]_{\mathcal{T}_1}  & & 
\tau_{0,-1}^{(0)} \ar[ru]_{\mathcal{T}_1}  &  & 
\tau_{1,-1}^{(0)} 
    }
\end{equation}

The structure of the transformations \eqref{Weylaction-tau} is 
thought of as the mirror image of this diagram.
In fact, $s_0^\mathrm{L}$ and $s_1^\mathrm{L}$ are 
reflections with respect to the vertical lines
which link $\tau^{(1)}_{0,0}$ and $\tau^{(1)}_{0,-1}$,
$\tau^{(0)}_{0,0}$ and $\tau^{(0)}_{0,1}$ respectively.
Also, $s_0^\mathrm{R}$ and $s_1^\mathrm{R}$ are 
reflections with respect to the horizontal lines
which link $\tau^{(1)}_{0,0}$ and $\tau^{(1)}_{-1,0}$,
$\tau^{(0)}_{0,0}$ and $\tau^{(0)}_{1,0}$ respectively.
Furthermore,
the composition $\pi\sigma$ represents the reflection
with respect to the slanted line 
through $\tau_{0,0}^{(0)}$ and $\tau_{0,0}^{(1)}$.

Below we list the Hirota bilinear equations of 
mPLR hierarchy.
From the differential equations for $\phi$ \eqref{defofU2n} we
get 
\begin{equation}
\label{bilinear-t1}
\begin{aligned}
 &D_{t_1} \tau^{(0)}_{m,n}\cdot\tau^{(1)}_{m,n} = 2\tau^{(1)}_{m-1,n}\tau^{(0)}_{m+1,n},
\end{aligned}
\end{equation}
and
\begin{equation}
\label{bilinear-t2t3}
\begin{aligned}
 &(D_{t_1}^2+2D_{t_2})\tau^{(0)}_{m,n} \cdot \tau^{(1)}_{m,n} 
= 0,
\\
 &(D_{t_1}^3 -4D_{t_3}) \tau^{(1)}_{m,n}\cdot\tau^{(0)}_{m,n} 
= 24\tau^{(0)}_{m-1,n}\tau^{(1)}_{m+1,n}.
\end{aligned}
\end{equation}
The differential equations for $q, r, \bar{q}, \bar{r}$ \eqref{defofU2n+1} imply
\begin{equation}
\label{6-siki-tau}
\begin{aligned}
&D_{t_1} \tau^{(0)}_{m,n+1}\cdot\tau^{(0)}_{m,n} 
= 2\tau^{(1)}_{m-1,n}\tau^{(1)}_{m,n},
\end{aligned}
\end{equation}
and the differential equations for $W_2$ and $\bar{W}_2$ \eqref{W2tnbibun} give rise to
\begin{align}
\label{D1D1bar00}
 &(D_{t_1}D_{\bar{t}_1} + 4)\tau^{(0)}_{m,n}\cdot\tau^{(0)}_{m,n}  
 = -8\tau^{(1)}_{m-1,n} \tau^{(1)}_{m,n-1}.
\end{align}
This relation is the 2D Toda equation.

We also have an algebraic equation.
The transformation $\mathcal{T}_1$ for $e^\phi$ is 
$$
\mathcal{T}_1(e^\phi) = s_0^\mathrm{R}s_0^\mathrm{L}\pi(e^\phi) 
= e^\phi-q\bar{r}e^{-\phi}.
$$
Rewriting this as a relation among $\tau$-functions,
we have
\begin{equation}
\begin{aligned}
 (\tau^{(0)}_{m,n})^2 &= \tau^{(1)}_{m,n}\tau^{(1)}_{m-1,n-1} 
- \tau^{(1)}_{m-1,n}\tau^{(1)}_{m,n-1}.
\end{aligned}
\end{equation}

We remark 
that the action of operators $\mathcal{T}_1, \mathcal{T}_2$ on
the linear operators of mPLR hierarchy are obtained by
$$
 \frac{\partial}{\partial t_n} - \mathcal{T}_i(B_n) 
= \mathcal{R}_i \left( \frac{\partial}{\partial t_n} - B_n \right) 
\mathcal{R}_i^{-1},
$$
where
\begin{align}
\label{zerozero}
&\mathcal{R}_1 = 
\begin{bmatrix}
 0 & -q \\ 0 & e^{2\phi}/(q\bar{r} - e^{2\phi})  
\end{bmatrix}\zeta^{-1/2}
+ \begin{bmatrix}
 1 & 0 \\ \bar{r}/(e^{2\phi}-q\bar{r} ) & 0 
\end{bmatrix}\zeta^{1/2},
\\
\label{ichizero}
&\mathcal{R}_2 = 
\begin{bmatrix}
  q/e^{2\phi}\bar{q} &  -q \\  0 & 0 
\end{bmatrix}\zeta^{-1/2} 
+ \begin{bmatrix}
 1 & 0 \\  -1/q & 0 
\end{bmatrix}\zeta^{1/2}.
\end{align}

\section{Similarity reduction}

\subsection{Similarity conditions for the wave functions}

Now we construct a similarity reduction of the mPLR hierarchy.
Let $\lambda \in \IC$ be a complex parameter and
set
$t_\lambda = (\lambda t_1, \lambda^2 t_2, \dots, \lambda^nt_n, \dots)$,
$\bar{t}_{\lambda^{-1}}
= (\lambda^{-1}\bar{t}_1, \lambda^{-2}\bar{t}_2, \dots, \lambda^{-n}\bar{t}_n, \dots)$.
We consider a one-parameter 
transformation for $g_{>0}$, $g_{<0}$ and $g_0$.
\begin{proposition}
\label{propSR}
If $g_{<0}$ and $g_{\ge 0} = g_0g_{>0}$ solve the Sato-Wilson equations \eqref{homogSW},
the functions
\begin{equation}
\label{hom-souji-1}
\begin{aligned}
 g_{<0}^\lambda(\zeta; t, \bar{t}) 
&:= \lambda^{\alpha H_0} g_{<0}(\lambda^{-1}\zeta; t_\lambda, \bar{t}_{\lambda^{-1}}) 
\lambda^{-\alpha H_0},
\\
 g_0^\lambda(t, \bar{t}) 
&:=
   \lambda^{\alpha H_0} g_0(t_\lambda, \bar{t}_{\lambda}) \lambda^{-\beta H_0},
\\
 g_{>0}^\lambda (\zeta; t, \bar{t}) 
&:= \lambda^{\beta H_0} 
g_{>0}(\lambda^{-1}\zeta; t_\lambda, \bar{t}_{\lambda^{-1}})
\lambda^{-\beta H_0}
\end{aligned}
\end{equation}
also solve \eqref{homogSW}.
\end{proposition}
This is verified by direct calculation.
Proposition \ref{propSR} shows 
that $g_{<0}, g_0$ and $g_{>0}$ have the scaling symmetry,
it thus makes sense 
to look only at functions with the 
following similarity conditions:
\begin{equation}
\label{similarity-reductions}
 g_{<0}^\lambda(\zeta, t, \bar{t}) = g_{<0}(\zeta, t, \bar{t}),
 \qquad
 g_0^\lambda(\zeta, t, \bar{t}) = g_0(\zeta, t, \bar{t}),
 \qquad
 g_{>0}^\lambda(\zeta, t, \bar{t}) = g_{>0}(\zeta, t, \bar{t}).
\end{equation}
The components of these matrices
satisfy, for example,
$e^{\phi(t_\lambda, \bar{t}_{\lambda^{-1}})} 
= \lambda^{\beta-\alpha}e^{\phi(t, \bar{t})}$,
\begin{alignat*}{2}
&\left\{
\begin{aligned}
   &q(t_\lambda, \bar{t}_{\lambda^{-1}}) = \lambda^{-2\alpha -1}q(t, \bar{t}),
\\ 
   &r(t_\lambda, \bar{t}_{\lambda^{-1}}) = \lambda^{2\alpha}r(t, \bar{t}),
\end{aligned}
\right.
&\qquad
&\left\{
\begin{aligned}
   &w_{21}(t_\lambda, \bar{t}_{\lambda^{-1}}) = \lambda^{-1}w_{21}(t, \bar{t}),
\\ 
   &w_{22}(t_\lambda, \bar{t}_{\lambda^{-1}}) = \lambda^{-1}w_{22}(t, \bar{t}), 
\end{aligned}
\right.
\\
&\left\{
\begin{aligned}
  &\bar{q}(t_\lambda, \bar{t}_{\lambda^{-1}}) = 
 \lambda^{-2\beta}\bar{q}(t, \bar{t}),
\\
  &\bar{r}(t_\lambda, \bar{t}_{\lambda^{-1}}) = 
 \lambda^{2\beta + 1}\bar{r}(t, \bar{t}),
\end{aligned}
\right.
&\qquad
&\left\{
\begin{aligned}  
&\bar{w}_{21}(t_\lambda, \bar{t}_{\lambda^{-1}}) 
= \lambda\bar{w}_{21}(t, \bar{t}), 
\\
  &\bar{w}_{22}(t_\lambda, \bar{t}_{\lambda^{-1}}) 
= \lambda\bar{w}_{22}(t, \bar{t}).
\end{aligned}
\right.
\end{alignat*}
The similarity conditions \eqref{similarity-reductions} can be expressed in
infinitesimal form:
\begin{align}
\label{hom-souji-2}
 \zeta\frac{\partial g_{<0}}{\partial \zeta}
&= [\alpha H_0, g_{<0}] 
+ \sum_n nt_n \frac{\partial g_{<0}}{\partial t_n}
- \sum_n n\bar{t}_n \frac{\partial g_{<0}}{\partial \bar{t}_n},
\\
\label{hom-sim-W0}
 0&= (\alpha H_0) g_0
+ \sum_n nt_n \frac{\partial g_0}{\partial t_n}
- \sum_n n\bar{t}_n \frac{\partial g_0}{\partial \bar{t}_n}
-g_0(\beta H_0),
\\
\label{hom-souji-2bar}
 \zeta\frac{\partial g_{>0}}{\partial \zeta}
&= [\beta H_0, g_{>0}] 
+ \sum_n nt_n \frac{\partial g_{>0}}{\partial t_n}
- \sum_n n\bar{t}_n \frac{\partial g_{>0}}{\partial \bar{t}_n}.
\end{align}
In particular, 
the variables $W_1$, $\bar{W}_1$, $W_2$ and $\bar{W}_2$ satisfy the 
following conditions
by virtue of the Sato-Wilson equations:
\begin{gather}
\label{hom-W1-sim}
\begin{aligned}
&\begin{bmatrix} 2\alpha+1 & 0 \\ 0 & -2\alpha  \end{bmatrix}W_1
= \sum_n nt_nU_{2n+1} + \sum_n n\bar{t}_n e^{2\Phi}\bar{U}_{2n-1},
\\
&\begin{bmatrix} -2\beta & 0 \\ 0 & 2\beta + 1  \end{bmatrix}\bar{W}_1
= \sum_n nt_n e^{-2\Phi}U_{2n-1} + \sum_n n\bar{t}_n\bar{U}_{2n+1}, 
\end{aligned}
\\
\label{hom-W2-sim}
\begin{aligned}
 W_2 &= \sum_n nt_n(U_{2n+2} + U_{2n+1}\Lambda W_1\Lambda^{-1})
+ \sum_n n\bar{t}_n(\bar{U}_{2n-2} + e^{2\Phi}\bar{U}_{2n-1}
\Lambda W_1 \Lambda^{-1}),
\\
 \bar{W}_2 &=
\sum_n nt_n (U_{2n-2} + e^{-2\Phi} U_{2n-1} \Lambda \bar{W}_1 \Lambda^{-1})
+\sum_n n\bar{t}_n (\bar{U}_{2n+2} + \bar{U}_{2n+1} \Lambda \bar{W}_1 \Lambda^{-1} ).
\end{aligned}
\end{gather}
We emphasize that the variables $W_2$ and $\bar{W}_2$ can
be described by differential polynomials of $q, r$ and $\bar{q}, \bar{r}$ (see
Proposition \ref{prop1}) under the similarity conditions 
\eqref{similarity-reductions}.
Furthermore, from the condition for $g_0$ \eqref{hom-sim-W0}, we obtain
\begin{equation}
\label{hom-W0bar-sim}
 \sum_n nt_n U_{2n} + \sum_n n\bar{t}_n\bar{U}_{2n}
= (\beta-\alpha) H_0.
\end{equation}

Under the similarity conditions \eqref{similarity-reductions},
the wave functions $Y = \Psi$ and $\bar{\Psi}$ \eqref{BAfunction} satisfy
\begin{equation}
\label{BAsim}
 Y(\lambda \zeta; t, \bar{t})
 = \lambda^{\alpha H_0} Y(\zeta; t_{\lambda}, \bar{t}_{\lambda^{-1}}),
\end{equation}
because of the relations $\Psi_0(\lambda \zeta; t)
= \Psi_0(\zeta; t_{\lambda})$,
$\bar{\Psi}_0(\lambda \zeta;  \bar{t})
 = \bar{\Psi}_0(\zeta;  \bar{t}_{\lambda^{-1}})$.
Formula \eqref{BAsim} implies the following linear equation:
\begin{align}
 \zeta\frac{\partial Y}{\partial \zeta} 
&= \alpha H_0Y 
+ \sum_n nt_n \frac{\partial Y}{\partial t_n}
- \sum_n n \bar{t}_n \frac{\partial Y}{\partial \bar{t}_n} 
\nonumber
\\
&= \left(
  \alpha H_0 + \sum_n nt_nB_n - \sum_n n \bar{t}_n (g_0 \bar{B}_n g_0^{-1}) 
\right) Y
\label{simBA}
\end{align}
for $Y = \Psi, \bar{\Psi}$.
From the compatibility condition of the 
linear equation \eqref{simBA} and \eqref{linearDE},
we get  an isomonodromic 
deformation equation \cite{FN}, \cite{JM2}:
\begin{equation}
\label{mpd}
 \frac{\partial A}{\partial t_i} = [B_i,A] + \zeta\frac{\partial B}{\partial \zeta},
\qquad
A=  \alpha H_0 + \sum_n nt_nB_n - \sum_n n \bar{t}_n (g_0 \bar{B}_n g_0^{-1}).
\end{equation}

\subsection{Affine Weyl group symmetry of the reduction parameters}

The affine Weyl group symmetry still exists after the similarity reduction.
In addition,
we have the following 
transformation for the parameters $\alpha$ and $\beta$ in
the similarity conditions \eqref{hom-souji-1}, \eqref{similarity-reductions}.
\begin{theorem}
The six basic transformations $\sigma, 
s_i^\mathrm{L}, s_i^\mathrm{R} (i=0,1)$ and $\pi$ induce
the following transformations of
the parameters $\alpha, \beta$ in
the similarity conditions \eqref{hom-souji-1}, \eqref{similarity-reductions}:
\begin{equation}
\label{table-parameter}
\begin{array}{c|cc}
  & \alpha & \beta 
\\ \hline
\sigma  & -\beta - \frac{1}{2} & -\alpha - \frac{1}{2}
\\
s_0^\mathrm{L}  & -\alpha - 1 & \beta
\\
s_1^\mathrm{L}  & -\alpha & \beta
 \\
 s_0^\mathrm{R}  &  \alpha & -\beta - 1
\\ 
 s_1^\mathrm{R}  & \alpha & -\beta
\\
 \pi  & -\alpha - \frac{1}{2} & -\beta - \frac{1}{2}
\end{array} 
\end{equation}
\end{theorem}
\begin{proof}
We rewrite the similarity condition \eqref{hom-souji-2} by
\begin{equation}
\label{zure}
\begin{aligned}
 \left[ \zeta\frac{\partial}{\partial \zeta} +\frac{1}{4}H_0, g_{<0}(t,\bar{t}) \right]
&= \left[(\alpha+\frac{1}{4}) H_0, g_{<0}(t,\bar{t}) \right] 
\\
&\quad
+ \sum_n nt_n \frac{\partial g_{<0}}{\partial t_n}(t,\bar{t})
- \sum_n n\bar{t}_n \frac{\partial g_{<0}}{\partial \bar{t}_n}(t,\bar{t}).
\end{aligned}
\end{equation}
If we let $\sigma$ \eqref{defofsigma} act on
the formula \eqref{zure}, 
then by using the relation $\left[ 
\zeta\frac{\partial}{\partial \zeta} +\frac{1}{4}H_0,  \Lambda^n \right]
= \frac{n}{2}\Lambda^n$,
we have 
\[
\begin{aligned}
 \left[ -\zeta\frac{d}{d \zeta} -\frac{1}{4}H_0, g_{>0}(\bar{t},t) \right]
&= \left[(\sigma(\alpha)+\frac{1}{4}) H_0, g_{>0}(\bar{t},t) \right] 
\\
&\quad
+ \sum_n nt_n \frac{\partial g_{>0}}{\partial \bar{t}_n}(\bar{t},t)
- \sum_n n\bar{t}_n \frac{\partial g_{>0}}{\partial t_n}(\bar{t},t).
\end{aligned}
\]
Comparing this equation with \eqref{hom-souji-2bar},
we have $\sigma(\alpha) = -\beta -\frac{1}{2}$.

Next we consider the actions of $s_i^\mathrm{L}$ and $s_i^\mathrm{R}$.
Since the generator of permutation $S_0=S_0(\zeta)$ \eqref{Weyls0s1} satisfies $S_0(\zeta)
= \lambda^{-H_0}S_0(\lambda^{-1}\zeta)$,
under the similarity conditions,
the matrices in the definition of the left action \eqref{Finfty} and 
the right action \eqref{GzeroGone} satisfy
\begin{align*}
&\left\{
\begin{aligned}
 &G_0^\mathrm{L}(\zeta;t,\bar{t})
= \lambda^{-(\alpha+1) H_0}
G_0^\mathrm{L}(\lambda^{-1}\zeta;t_\lambda,\bar{t}_{\lambda^{-1}})
\lambda^{-\alpha H_0},
\\
 &\bar{G}_0^\mathrm{L}(\zeta;t,\bar{t})
=\lambda^{\beta H_0}
\bar{G}_0^\mathrm{L}(\lambda^{-1}\zeta;t_\lambda,\bar{t}_{\lambda^{-1}})
\lambda^{-\beta H_0},
\end{aligned}
\right.
\\
&\left\{
\begin{aligned}
 &G_1^\mathrm{L}(t,\bar{t}) 
=\lambda^{-\alpha H_0}
G_1^\mathrm{L}(t_\lambda,\bar{t}_{\lambda^{-1}})
\lambda^{-\alpha H_0},
\\
 &\bar{G}_1^\mathrm{L}(t,\bar{t}) 
=\lambda^{\beta H_0}
\bar{G}_1^\mathrm{L}(t_\lambda,\bar{t}_{\lambda^{-1}})
\lambda^{-\beta H_0},
\end{aligned}
\right.
\\
&\left\{
\begin{aligned}
 &G_0^\mathrm{R}(\zeta;t,\bar{t})
=\lambda^{\alpha H_0}
G_0^\mathrm{R}(\lambda^{-1}\zeta;t_\lambda,\bar{t}_{\lambda^{-1}})
\lambda^{-\alpha H_0},
\\
 &\bar{G}_0^\mathrm{R}(\zeta;t,\bar{t})
= \lambda^{-(\beta+1) H_0}
\bar{G}_0^\mathrm{R}(\lambda^{-1}\zeta;t_\lambda,\bar{t}_{\lambda^{-1}})
 \lambda^{-\beta H_0},
\end{aligned}
\right.
\\
&\left\{
\begin{aligned}
 &
G_1^\mathrm{R}(t,\bar{t})
 = \lambda^{\alpha H_0} 
G_1^\mathrm{R}(t_\lambda,\bar{t}_{\lambda^{-1}})
 \lambda^{-\alpha H_0},
\\
 &\bar{G}_1^\mathrm{R}(t,\bar{t}) 
= \lambda^{-\beta H_0}
\bar{G}_1^\mathrm{R}(t_\lambda,\bar{t}_{\lambda^{-1}})
\lambda^{-\beta H_0}. 
\end{aligned}
\right.
\end{align*}
Then, for example, the actions of $s_0^\mathrm{L}$ on the matrices 
$g_{<0}$, $g_0$ and $g_{>0}$ with the
similarity condition \eqref{similarity-reductions} are
\begin{equation}
\label{hom-souji-Weyl}
\begin{aligned}
&s_0^{\mathrm{L}}(g_{<0})(\zeta; t, \bar{t}) 
= \lambda^{-(\alpha+1) H_0} 
s_0^{\mathrm{L}}(g_{<0})(\lambda^{-1}\zeta; t_\lambda, \bar{t}_{\lambda^{-1}}) 
\lambda^{(\alpha+1) H_0},
\\
&s_0^{\mathrm{L}}(g_0)(t, \bar{t}) =
   \lambda^{-(\alpha+1) H_0} 
 s_0^\mathrm{L}(g_0)(t_\lambda, \bar{t}_{\lambda^{-1}}) 
   \lambda^{-\beta H_0},
\\
&s_0^\mathrm{L}(g_{>0}^\lambda) (\zeta; t, \bar{t}) 
= \lambda^{\beta H_0} 
 s_0^\mathrm{L}(g_{>0})(\lambda^{-1}\zeta; t_\lambda, \bar{t}_{\lambda^{-1}}) 
\lambda^{-\beta H_0}.
\end{aligned}
\end{equation}
This relation imposes the action of $s_0^\mathrm{L}$ on $\alpha$ and $\beta$.
Other transformations can be computed in the same way.

The action of $\pi$ is obtained by the relation
$P(\lambda \zeta) = \lambda^{-\frac{1}{2}H_0}P(\zeta)
= P(\zeta) \lambda^{\frac{1}{2}H_0}$.
\end{proof}

\begin{corollary}
The shift operators on the $A_1^{(1)}$ root- and weight-lattice 
act on the parameters $\alpha, \beta$ by the following table:
$$
\begin{array}{c|cc}
  & \alpha & \beta 
\\ \hline
T_\mathrm{L} = s_0^\mathrm{L}s_1^\mathrm{L} &  \alpha + 1 & \beta 
\\
T_\mathrm{R} = s_0^\mathrm{R}s_1^\mathrm{R} & \alpha  & \beta + 1
\\
\mathcal{T}_1 
= s_0^\mathrm{R}s_0^\mathrm{L}\pi & \alpha + \frac{1}{2} & \beta + \frac{1}{2}
\\
\mathcal{T}_2 
= s_1^\mathrm{R}s_0^\mathrm{L}\pi & \alpha + \frac{1}{2} & \beta - \frac{1}{2}
\end{array}
$$
\end{corollary}

\section{The third Painlev\'e equation}

\subsection{Deriving the third Painlev\'e equation}

In this section,
we use only the independent variables $t_1$, $\bar{t}_1$ 
and the dependent variables $\bar{q}, \bar{r}$.
So from the mPLR equations \eqref{mPLR} 
we eliminate $q, r, e^\phi$ and obtain:
\begin{equation}
\label{qbarrbar}
\left\{
\begin{aligned}
 \frac{\partial^2 \bar{q}}{\partial t_1\partial \bar{t}_1}
 &= 4\bar{q} + 4\bar{q}\bar{r} \frac{\partial \bar{q}}{\partial t_1},
 \\
 \frac{\partial^2  \bar{r}}{\partial t_1\partial \bar{t}_1}
 &= 4\bar{r} -4\bar{q}\bar{r}  \frac{\partial \bar{r}}{\partial t_1}.
\end{aligned}
\right.
\end{equation}
Furthermore, from the
similarity conditions  
\eqref{hom-sim-W0},
\eqref{hom-souji-2bar},
we obtain
\begin{equation}
\label{similarityPIII}
\left\{
\begin{aligned}
 &\frac{\partial \bar{q}}{\partial \bar{t}_1}
= \frac{t_1}{\bar{t}_1}
  \frac{\partial \bar{q}}{\partial t_1}
 +\frac{2\beta}{\bar{t}_1} \bar{q},
\\
 &\frac{\partial \bar{r}}{\partial \bar{t}_1}
= \frac{t_1}{\bar{t}_1}
  \frac{\partial \bar{r}}{\partial t_1}
 -\frac{2\beta + 1}{\bar{t}_1} \bar{r},  
\end{aligned}
\right.
\qquad
4\bar{q}\bar{r} = \frac{2(\beta - \alpha)}{\bar{t}_1}
 +\frac{t_1}{\bar{t}_1}
 \frac{\partial \bar{q}}{\partial t_1}
 \frac{\partial \bar{r}}{\partial t_1}.
\end{equation}
The system of equations \eqref{qbarrbar}, \eqref{similarityPIII} is 
equivalent to the third Painlev\'e equation.
In fact,
by differentiating the first couple of equations in \eqref{similarityPIII} 
with respect to $t_1$ and eliminating $\displaystyle 
\frac{\partial \bar{q}}{\partial t_1 \partial \bar{t}_1},
\frac{\partial \bar{r}}{\partial t_1 \partial \bar{t}_1}$ and
$\bar{q}\bar{r}$ by 
\eqref{qbarrbar}, 
\eqref{similarityPIII},
we have
\begin{equation}
\label{qbarrbart1-2}
\left\{
\begin{aligned}
 \frac{\partial^2 \bar{q}}{\partial t_1^2}
=& \frac{4\bar{t}_1}{t_1}\bar{q} 
 -\frac{2\alpha + 1}{t_1}\frac{\partial \bar{q}}{\partial t_1} 
 +\left( \frac{\partial \bar{q}}{\partial t_1} \right)^2
 \frac{\partial \bar{r}}{\partial t_1}, 
\\
 \frac{\partial^2 \bar{r}}{\partial t_1^2}
=& \frac{4\bar{t}_1}{t_1}\bar{r} 
 +\frac{2\alpha}{t_1}\frac{\partial \bar{r}}{\partial t_1} 
 -\frac{\partial \bar{q}}{\partial t_1} 
\left( \frac{\partial \bar{r}}{\partial t_1} \right)^2.
\end{aligned}
\right.
\end{equation}
Now we define variables $y$ and $z$ by
\begin{equation}
\label{PIII-yz}
y = \frac{t_1}{\bar{q}}\frac{\partial \bar{q}}{\partial t_1},
\qquad
z = \frac{\bar{q}}{2}\frac{\partial \bar{r}}{\partial t_1} 
\end{equation}
and put $s := 2t_1$,
after which the system of equations \eqref{qbarrbart1-2} can be rewritten as
\begin{equation}
\label{painleve3}
\left\{
\begin{aligned}
 s\frac{dy}{ds}
 &= 2y^2z -y^2
 -2\alpha y + 2s\bar{t}_1,
\\
 s\frac{dz}{ds}
 &= -2yz^2
+ 2yz + 2\alpha z + \beta - \alpha.
\end{aligned}
\right.
\end{equation}
This is the Hamiltonian equation $\frac{dy}{ds} = \frac{\partial h}{\partial z}$,
$\frac{dz}{ds} = -\frac{\partial h}{\partial y}$ 
with the Hamiltonian
\begin{equation}
\label{P3hamiltonian}
 sh = y^2z^2 - y^2z -2\alpha yz + 2s\bar{t}_1z +(\alpha-\beta)y
\end{equation}
and the variable $y$ satisfies 
\begin{equation}
\label{P3prime}
\frac{d^2y}{ds^2} 
= \frac{1}{y}\left(\frac{dy}{ds}\right)^2
- \frac{1}{s}\frac{dy}{ds}
+ \frac{y^2}{s^2}(y+2\beta)
 +\frac{2\bar{t}_1(2\alpha+1)}{s}
 -\frac{4\bar{t}_1^2}{y}.
\end{equation}
When we put $\bar{t}_1 = 1/2$,
we have 
Painlev\'e III$'$ \eqref{PIII'} 
with the parameters $c_1 = 8\beta$,
$c_2 = 4(2\alpha+1)$,
$c_3 = 4$, $c_4 = -4$.
We remark that if we express
the Hamiltonian $h$ \eqref{P3hamiltonian} in the 
variable of the mPLR hierarchy,
we get
$$
\begin{aligned}
 h &= 2t_1q^2r^2 + 2\alpha qr
+2\bar{t}_1(q\bar{r}e^{-2\phi} + \bar{q}re^{2\phi})
\\
 &= -w_{21} + \bar{t}_1
 = \frac{1}{2}\frac{\partial}{\partial t_1}\log\tau_{00}^{(0)} + \bar{t}_1
\end{aligned}
$$
by the similarity condition \eqref{hom-W2-sim}.

The system of equations \eqref{qbarrbart1-2} is obtained
by the compatibility condition
of the linear equations
$$
 \zeta\frac{\partial \Psi}{\partial \zeta} = A\Psi,
 \qquad
 \frac{\partial \Psi}{\partial t_1} = B_1\Psi.
$$
with the coefficient matrices 
\begin{align}
 A &= -\alpha H_0 + t_1B_1  - \bar{t}_1\bar{B}_1
\nonumber \\
 &= \begin{bmatrix} -\bar{t}_1 & 2\bar{t}_1\bar{q}e^{2\phi} \\ 
0 & \bar{t}_1 \end{bmatrix}\zeta^{-1}
+ \begin{bmatrix} 
 -\alpha + 2t_1qr & -2t_1q \\ 
-2\bar{t}_1\bar{r}e^{-2\phi} & \alpha -2t_1qr \end{bmatrix}
+ \begin{bmatrix} t_1 & 0 \\ 2t_1r & -t_1 \end{bmatrix}\zeta,
\end{align}
and $B_1$ \eqref{B1},
see \eqref{mpd}.
We believe this is a new Lax formation of the third Painlev\'e equation.
Notice that the relations
\begin{equation}
\label{sim-yz-relation}
y =-2t_1\frac{q}{\bar{q}e^{2\phi}},
\qquad
z = \bar{q}re^{2\phi},
\qquad
yz = -2t_1qr = 2\bar{t}_1\bar{q}\bar{r} + \alpha - \beta
\end{equation}
hold.

\subsection{Bilinear equations}

The canonical coordinates of the third Painlev\'e 
equation $y$ and $z$ are described in the tau function as follows:
\begin{align}
\label{tau-to-yz}
y = t_1\frac{\partial}{\partial t_1}\log \frac{\tau^{(1)}_{0,-1}}{\tau^{(1)}_{0,0}}
= -2t_1\frac{\tau^{(0)}_{0,0}\tau^{(0)}_{1,0}}{\tau^{(1)}_{0,0}\tau^{(1)}_{0,-1}},
\qquad
z = -\frac{\tau^{(1)}_{0,-1}\tau^{(1)}_{-1,0}}{(\tau^{(0)}_{0,0})^2}.
\end{align}
\begin{proposition}
\label{prop4}
The four $\tau$-functions $\tau^{(0)}_{0,0}, \tau^{(0)}_{1,0}, 
\tau^{(1)}_{0,0}$, and $\tau^{(1)}_{0,-1}$ (see the diagram \eqref{weight-lattice})
in $\eqref{tau-to-yz}$ with
the similarity condition 
satisfy
the following relations:
\begin{align}
\label{tau-1}
t_1D_{t_1}\tau^{(0)}_{0,0}\cdot \tau^{(0)}_{1,0}
&= 2\bar{t}_1\tau^{(1)}_{0,0}\tau^{(1)}_{0,-1}
+(2\alpha + 1)\tau^{(0)}_{0,0}\tau^{(0)}_{1,0},
\\
 t_1^2D_{t_1}^2\tau^{(0)}_{0,0}\cdot \tau^{(0)}_{1,0}
&= -2t_1(\tau^{(0)}_{0,0})'\tau_{0,0}^{(0)}
+4(\alpha + 1) \bar{t}_1 \tau^{(1)}_{0,0}\tau^{(1)}_{0,-1}
\label{q''-tau-2}
\\
&\quad
+ (2\alpha + 1)(2\alpha + 2)
\tau^{(0)}_{0,0}\tau^{(0)}_{1,0},
\nonumber
\\
 t_1D_{t_1}\tau^{(1)}_{0,0}\cdot\tau^{(1)}_{0,-1}
 &= 2t_1\tau^{(0)}_{0,0}\tau^{(0)}_{1,0},
\label{rho-1}
\\
\label{qbar''-tau}
 t_1^2D_{t_1}^2 \tau^{(1)}_{0,0}\cdot \tau^{(1)}_{0,-1}
&= -2t_1(\tau^{(1)}_{0,0})'\tau^{(1)}_{0,-1}
 + 2(2\beta+1)t_1\tau^{(0)}_{0,0}\tau^{(0)}_{1,0}.
\end{align}
\end{proposition}
\begin{proof}
The equation \eqref{rho-1} is nothing but \eqref{6-siki-tau} (apply $\mathcal{T}_2$).
From the similarity condition \eqref{hom-W1-sim} we have
\begin{equation}
\label{q'sim}
 t_1\frac{\partial q}{\partial t_1} - \bar{t}_1\frac{\partial q}{\partial \bar{t}_1}
 = -(2\alpha + 1)q.
\end{equation}
Translating this into the $\tau$-functions,
we get \eqref{tau-1}.

The equations \eqref{q''-tau-2} and \eqref{qbar''-tau} are obtained
by differentiating the equation \eqref{q'sim} with respect to $t_1$,
together with the equation
\begin{equation}
2t_1 \tau^{(0)}_{1,0} \tau^{(1)}_{-1,0}
+ 2\bar{t}_1 \tau^{(0)}_{0,1} \tau^{(1)}_{0,-1} = 
(\alpha - \beta)\tau^{(0)}_{0,0}\tau^{(1)}_{0,0},
\label{alpha-beta}
\end{equation}
which is obtained from the similarity condition for
$e^\phi$ \eqref{hom-W0bar-sim},
and
\begin{equation}
\label{souji-00}
\begin{aligned}
 \frac{D_1^2\tau^{(0)}_{0,0}\cdot \tau^{(0)}_{0,0}}{(\tau^{(0)}_{0,0})^2}
 =& -\frac{2}{t_1}\frac{(\tau^{(0)}_{0,0})'}{\tau^{(0)}_{0,0}}
 -\frac{8\bar{t}_1}{t_1} \frac{\tau^{(1)}_{0,-1}\tau^{(1)}_{-1,0}}
{(\tau^{(0)}_{00})^2} -\frac{4\bar{t}_1}{t_1},
\\
 \frac{D_1^2\tau^{(1)}_{0,0}\cdot \tau^{(1)}_{0,0}}{(\tau^{(1)}_{0,0})^2}
 =& -\frac{2}{t_1}\frac{(\tau^{(1)}_{0,0})'}{\tau^{(1)}_{0,0}}
 -\frac{8\bar{t}_1}{t_1} \frac{\tau^{(0)}_{1,0}\tau^{(0)}_{0,1}}
{(\tau^{(1)}_{0,0})^2} -\frac{4\bar{t}_1}{t_1},
\end{aligned}
\end{equation}
which is obtained from the similarity condition for $W_2$ \eqref{hom-W2-sim}.
\end{proof}
The bilinear form for the third Painlev\'e equation \eqref{P3prime}
\begin{equation}
\label{P3bilinear}
\begin{aligned}
 &t_1^2 \left( 
\frac{D_1^2\tau_{0,0}^{(0)}\cdot \tau_{1,0}^{(0)}}
{\tau_{0,0}^{(0)}\tau_{1,0}^{(0)}}
 - \left( \frac{D_1 \tau_{0,0}^{(0)}\cdot \tau_{1,0}^{(0)}}
{\tau_{0,0}^{(0)}\tau_{1,0}^{(0)}} \right)^2
 -\frac{D_1^2\tau^{(1)}_{0,0}\cdot \tau^{(1)}_{0,-1}}
{\tau^{(1)}_{0,0}\tau^{(1)}_{0,-1}}
 + \left( \frac{D_1 \tau^{(1)}_{0,0}\cdot \tau^{(1)}_{0,-1}}
{\tau^{(1)}_{0,0}\tau^{(1)}_{0,-1}} \right)^2 \right)
\\
&\quad + t_1\left(
\frac{D_1 \tau_{0,0}^{(0)}\cdot \tau^{(1)}_{0,-1}}
{\tau_{0,0}^{(0)}\tau^{(1)}_{0,-1}} 
 + \frac{D_1 \tau_{1,0}^{(0)}\cdot \tau^{(1)}_{0,0}}
{\tau_{1,0}^{(0)}\tau^{(1)}_{0,0}}  \right)
\\
& =4t_1^2\left( \frac{\tau_{0,0}^{(0)}\tau_{1,0}^{(0)}}
{\tau^{(1)}_{0,0}\tau^{(1)}_{0,-1}} \right)^2
- 4\beta t_1 \frac{\tau_{0,0}^{(0)}\tau_{1,0}^{(0)}}{\tau^{(1)}_{0,0}\tau^{(1)}_{0,-1}}
- 2(2\alpha + 1) \bar{t}_1 \frac{\tau^{(1)}_{0,0}\tau^{(1)}_{0,-1}}
{\tau_{0,0}^{(0)}\tau_{1,0}^{(0)}}
- 4\bar{t}_1^2\left( \frac{\tau^{(1)}_{0,0}\tau^{(1)}_{0,-1}}
{\tau_{0,0}^{(0)}\tau_{1,0}^{(0)}} \right)^2
\end{aligned}
\end{equation}
can be derived from Proposition \ref{prop4}.

\subsection{Affine Weyl group symmetry of P$_{\mathrm{III}'}$}

The action of the six basic transformations 
on the parameters $\alpha$ and $\beta$ is given by \eqref{table-parameter}.
We now consider the action on the
variables $y$ and $z$ \eqref{PIII-yz},
\eqref{sim-yz-relation}.
The similarity condition
discussed above leads to the following assumption 
for the solutions of the mPLR equation:
\begin{equation}
\label{simqrbar1->2}
 \bar{q}(t_1, \bar{t}_1) = (2\bar{t}_1)^{2\beta}\bar{Q}(\xi),
 \qquad
 \bar{r}(t_1, \bar{t}_1) = (2\bar{t}_1)^{-2\beta-1} \bar{R}(\xi),
\end{equation}
where $\xi = 2t_1\bar{t}_1$.
The functions $\bar{Q}(\xi)$ and $\bar{R}(\xi)$ satisfy 
\begin{equation}
\label{QbarRbar}
\left\{
\begin{aligned}
 \frac{d^2 \bar{Q}}{d\xi^2}
=& \frac{\bar{Q}}{\xi} 
- \frac{2\alpha + 1}{\xi}\frac{d \bar{Q}}{d\xi} 
 + 2\left( \frac{d \bar{Q}}{d\xi} \right)^2
 \frac{d\bar{R}}{d\xi},
\\
 \frac{d^2 \bar{R}}{d\xi^2}
=& \frac{\bar{R}}{\xi} 
 +\frac{2\alpha}{\xi}\frac{d\bar{R}}{d\xi} 
 -2\frac{d\bar{Q}}{d\xi} 
\left( \frac{d\bar{R}}{d\xi} \right)^2.
\end{aligned}
\right.
\end{equation}
Equations \eqref{QbarRbar} are a rewriting of the equations \eqref{qbarrbart1-2}.
Then the variables $y$ and $z$ 
can be expressed
\begin{equation}
\label{sim-yz}
 y(t_1,\bar{t}_1) 
= \frac{2t_1\bar{t}_1}{\bar{Q}(2t_1\bar{t}_1)}
\frac{d\bar{Q}}{d\xi}(2t_1\bar{t}_1),
 \qquad
 z(t_1,\bar{t}_1) = 
\bar{Q}(2t_1\bar{t}_1)\frac{d\bar{R}}{d\xi}(2t_1\bar{t}_1).
\end{equation}
Hereafter we regard the variables $y$ and $z$ as a function 
of $\xi = 2t_1t_2$.

The actions of $\sigma$ and $\pi$ are given by 
\begin{align}
\label{sim-pi-sigma}
&\left\{
\begin{aligned}
 \sigma(y) &= \frac{2\xi}{y},
 \\
 \sigma(z) &= \frac{y(\alpha - \beta -yz)}{2\xi},
\end{aligned}
\right.
\qquad
\left\{
\begin{aligned}
 \pi(y) 
&= \frac{2\xi z}{\beta-\alpha+yz},
 \\
 \pi(z) 
&=\frac{y(\alpha - \beta - yz)}{2\xi}
\end{aligned}
\right.
\end{align}
because of the
definition and the relations \eqref{sim-yz-relation}.
We note that 
the composition $\pi \sigma$ is 
simpler than $\pi$: 
$$
\left\{
\begin{aligned}
 \pi\sigma(y) 
&= y + \frac{\beta-\alpha}{z},
 \\
 \pi\sigma(z) &= z.
\end{aligned}
\right.
$$
The other transformations 
act on $y$ and $z$ as follows:
\begin{align*}
&\left\{
\begin{aligned}
s_0^{\mathrm{L}}(y)(\xi) &=  y(-\xi),
\\
s_0^{\mathrm{L}}(z)(\xi) &= z(-\xi) 
- \frac{2\alpha+1}{y(-\xi)} + \frac{2\xi}{y^2(-\xi)},
\end{aligned} 
\right.
\\
&\left\{
\begin{aligned}
s_1^{\mathrm{L}}(y)(\xi) 
&=  \frac{2\xi z(-\xi)(1-z(-\xi))}
{\beta - \alpha + 2\alpha z(-\xi) -y(-\xi)z(-\xi) - y(-\xi)z^2(-\xi)},
 \\
s_1^{\mathrm{L}}(z)(\xi) &=  1-z(-\xi),
\end{aligned} 
\right.
\\
&\left\{
\begin{aligned}
s_0^{\mathrm{R}}(y)(\xi)
&= \frac{2\xi + y(-\xi)(\beta - \alpha + y(-\xi)z(-\xi))}
{\beta - \alpha + y(-\xi)z(-\xi) -(2\beta + 1) 
+ \frac{2\xi}{\beta - \alpha + y(-\xi)z(-\xi)}},
\\
s_0^{\mathrm{R}}(z)(\xi) 
&= z(-\xi) -\frac{(2\beta + 1)z(-\xi)}{\beta - \alpha + y(-\xi)z(-\xi)}
+ 2\xi\left(\frac{z(-\xi)}{\beta - \alpha + y(-\xi)z(-\xi)}\right)^2,
\end{aligned} 
\right.
\\
&\left\{
\begin{aligned}
s_1^{\mathrm{R}}(y)(\xi) &=  -y(-\xi),
 \\
s_1^{\mathrm{R}}(z)(\xi) &= 1-z(-\xi).
\end{aligned} 
\right.
\end{align*}
These are shown by the relation \eqref{sim-yz-relation} and
the similarity conditions.
For instance,
\begin{align*}
s_0^\mathrm{L}(r) 
= -\hat{q}
&= \frac{1}{2}\frac{\partial q}{\partial t_1} - q^2r
\\
&= q^2r + \frac{1}{t_1}\left( \alpha + \frac{1}{2} \right)q
+ \frac{\bar{t}_1}{t_1} \bar{q}e^{2\phi} 
\end{align*}
by the similarity condition \eqref{q'sim}.

Here we give a correspondence between the birational transformations
discussed in \cite{TOS} and in this paper.
In \cite{TOS},
the $A_1^{(1)}$ root systems $\{ \alpha_0, \alpha_1 \}$ and $\{
\beta_0, \beta_1 \}$ are expressed as
$$
(\alpha_0, \alpha_1)
= (1+\alpha-\beta, -\alpha + \beta),
\qquad
(\beta_0, \beta_1)
= (1+\alpha+\beta,  -\alpha - \beta).
$$
So the 
action of the extended affine Weyl group on these root systems
are realized as follows:
$$
 \begin{array}{c|cccc}
 & \alpha_0 & \alpha_1 & \beta_0 & \beta_1 
\\ 
\hline
s_1^\mathrm{R}s_0^\mathrm{R}
s_0^\mathrm{L}s_1^\mathrm{L}\pi\sigma
 & -\alpha_0 & \alpha_1 + 2\alpha_0 & \beta_0 & \beta_1 
\\ 
\pi \sigma  & \alpha_0 + 2\alpha_1 & -\alpha_1 & \beta_0 & \beta_1 
\\ 
s_0^\mathrm{R}s_0^\mathrm{L}\pi\sigma 
& \alpha_0 & \alpha_1 & -\beta_0 & \beta_1 + 2\beta_0 
\\ 
s_1^\mathrm{R}s_1^\mathrm{L}\pi\sigma 
& \alpha_0 & \alpha_1 & \beta_0 + 2\beta_1 & -\beta_1 
\\ 
\hline
s_1^\mathrm{R} s_0^\mathrm{L} \sigma & 
\alpha_1 & \alpha_0 & \beta_0 & \beta_1 
\\ 
\sigma & \alpha_0 & \alpha_1 & \beta_1 & \beta_0 
\\ 
\hline
s_1^\mathrm{R} & \beta_0 & \beta_1 & \alpha_0 & \alpha_1 
\\ 
s_0^\mathrm{L} & \beta_1 & \beta_0 & \alpha_1 & \alpha_0 
\end{array}
$$
The action on the variable $y$ and $z$ are also given by 
calculating these compositions.

\section{The fourth and the second Painlev\'e equations}

Finally,
we discuss the relation between 
the mPLR hierarchy and the fourth and the second Painlev\'e equations.

\subsection{Painlev\'e IV}

As we have shown in \cite{KK1},
if we use the variables $t_1$ and $t_2$,
we can obtain the fourth Painlev\'e equation.
Here we give additional information about $\tau$-functions and Hamiltonian functions.
In this case, the
similarity conditions for $W_1$ \eqref{hom-W1-sim} and $g_0$ \eqref{hom-W0bar-sim} are
\begin{align*}
 &t_1q + 2t_2\left( \frac{q''}{2} - 2q^2r' - 4q^3r^2 \right)
= -(2\alpha + 1) q, 
\\
&2t_1qr + 2t_2\left( q'r-qr'-2q^2r^2 \right)
= \beta - \alpha.
\end{align*}
So, the bilinear equations of these relations are
\begin{align}
&\left\{
\begin{aligned}
&(t_1D_{t_1} -t_2D_{t_1}^2 
-2\alpha -1 )\tau^{(0)}_{0,0} \cdot \tau^{(0)}_{1,0} = 0,
\\
&(t_1D_{t_1} - t_2D_{t_1}^2 + \alpha - \beta)
\tau^{(0)}_{0,0} \cdot \tau^{(1)}_{0,0} = 0.
\end{aligned}
\right.
\end{align}
When we put $t_2 = 1/2$ and
\begin{align*}
 y &= 2qr 
= -\frac{\tau^{(0)}_{1,0}\tau^{(1)}_{-1,0}}{\tau^{(0)}_{0,0}\tau^{(1)}_{0,0}}
= \frac{(\tau^{(1)}_{0,0})'}{\tau^{(1)}_{0,0}} 
- \frac{(\tau^{(0)}_{0,0})'}{\tau^{(0)}_{0,0}},
\\
 z &= 2t_1 + \frac{q'}{q}
 = 2t_1 + \frac{(\tau^{(0)}_{1,0})'}{\tau^{(0)}_{0,0}} 
- \frac{(\tau^{(0)}_{0,0})'}{\tau^{(0)}_{0,0}},
\end{align*}
we obtain 
$$
\left\{
\begin{aligned}
 y' &= y(2z-y-2t_1) +2(\alpha - \beta),
 \\
 z' &= z(2y-z+2t_1) - 4\beta.
\end{aligned}
\right.
$$
This system of equations is a Hamiltonian system
with polynomial Hamiltonian 
\begin{align*}
 H  = yz(z-y-2t_1) + 2(\beta - \alpha)z+4\beta y.
\end{align*}
By the similarity condition for $W_2$ \eqref{hom-W2-sim}, this Hamiltonian function is
expressed by
$$
 H=4t_1(\alpha -\beta)
 -8w_{21}.
$$

\subsection{Painlev\'e II}

We believe this result is new.
There are two steps to get the second Painlev\'e equation.
First,
we consider the similarity condition
with respect to $t_1$ and $t_3$.
The $t_3$-derivatives of $\bar{q}$ (not $q$) and $\phi$ are
\begin{equation*}
\begin{aligned}
 \frac{\partial \bar{q}}{\partial t_3} 
&= e^{-2\phi}u_{51} 
= -e^{-2\phi}\left(\frac{q''}{2} -2q^2r' - 4q^3r^2 \right),
\\
 \frac{\partial \phi}{\partial t_3}
&= u_{61} 
= \frac{1}{2}(q''r+qr''-q'r')-4q^3r^3.
\end{aligned}
\end{equation*}
So, the similarity conditions for $\bar{W}_1$ \eqref{hom-W1-sim} and
for $g_0$ \eqref{hom-W0bar-sim} are
\begin{align}
\label{sim-qbarP2}
 t_1\frac{\partial \bar{q}}{\partial t_1}
+3t_3\frac{\partial \bar{q}}{\partial t_3} 
=& -e^{-2\phi}\left( 2t_1q +3t_3\left( \frac{q''}{2} -2q^2r' - 4q^3r^2 \right)\right)
= -2\beta\bar{q},
\\
\label{sim-phiP2}
t_1\frac{\partial \phi}{\partial t_1}
+3t_3\frac{\partial \phi}{\partial t_3}
=& 2t_1qr + 3t_3\left( \frac{1}{2}(q''r+qr''-q'r')-4q^3r^3 \right)
= \beta - \alpha.
\end{align}
Hence, if we put $t_3 = 1/3$ and
\begin{equation}
\begin{aligned}
y &= -\frac{q'}{2q} 
= \frac{1}{2}\left(  \frac{(\tau^{(0)}_{0,0})'}{\tau^{(0)}_{00}} 
- \frac{(\tau^{(0)}_{1,0})'}{\tau^{(0)}_{1,0}} \right),
\\
z &
= 2q\hat{r} 
= -qr' - 2q^2r^2
= 2\frac{\tau^{(0)}_{1,0}\tau^{(0)}_{-1,0}}{(\tau^{(0)}_{0,0})^2},
\end{aligned}
\end{equation}
then the equations \eqref{sim-qbarP2} and \eqref{sim-phiP2} can
be rewritten as
\begin{equation}
\left\{
\begin{aligned}
y' &= 2y^2 + 2z + 2t_1 -2\beta \frac{e^{2\phi}\bar{q}}{q},
\\
z' &= -4yz +2\alpha +2\beta\left(2e^{2\phi}\bar{q}r - 1\right).
\end{aligned}
\right.
\end{equation}
Secondly,
if we put $\beta = 0$ and $s=2t_1$,
we have
\begin{equation}
\left\{
\begin{aligned}
\frac{dy}{ds}
=& y^2 + z + \frac{s}{2},
\\
\frac{d z}{d s}
=& -2yz +\alpha.
\end{aligned}
\right.
\end{equation}
This is a Hamiltonian form of the second Painlev\'e euqation
with Hamiltonian
$$
 H = \frac{z^2}{2}  + y^2z + \frac{t}{2}z -\alpha y.
$$
We remark that the equation \eqref{sim-qbarP2} can be described as
\begin{equation}
\label{P2bi-1}
\left( 2t_1 + \frac{3}{2}t_3D_{t_1}^2 \right) 
\tau^{(0)}_{0,0}\cdot \tau^{(0)}_{1,0} 
= 2\beta \tau^{(1)}_{0,0}\tau^{(1)}_{0,-1},
\end{equation}
and the similarity condition for $q$:
$$
 t_1\frac{\partial q}{\partial t_1}
+3t_3\frac{\partial q}{\partial t_3}
= t_1q' + 3t_3\left( \frac{q'''}{4} -3qq'r'-6q^2r^2q' \right)
= -(2\alpha + 1) q 
$$
is expressed as
\begin{equation}
\label{P2bi-2}
\left(
t_1D_{t_1} + \frac{3}{4}t_3D_{t_1}^3 
-2\alpha -1 \right)
\tau^{(0)}_{0,0} \cdot \tau^{(0)}_{1,0} = 0.
\end{equation}
The system of equations \eqref{P2bi-1},
\eqref{P2bi-2} with $\beta=0$ is a bilinear form 
of the second Painlev\'e equation.

\section*{Acknowledgments}

The author would like to thank Professors 
S.~Kakei,   
K.~Takasaki, 
T.~Tsuda,
and R.~Willox 
for their
several helpful comments.
This work was partially supported by 
the Global COE Program of University of Tokyo:
The research and training center for new development in mathematics.


\end{document}